\documentclass[11pt,letterpaper]{article}

\usepackage{amsmath,amsthm,amssymb}
\usepackage{thm-restate}
\usepackage[utf8]{inputenc}
\usepackage[dvipsnames]{xcolor}
\usepackage[shortlabels]{enumitem}
\usepackage[hypertexnames=false,colorlinks=true,urlcolor=Blue,citecolor=Blue,linkcolor=BrickRed]{hyperref}
\usepackage[capitalise]{cleveref}
\usepackage[OT4]{fontenc}
\usepackage[ruled]{algorithm2e}
\usepackage{tikz}
\usepackage[margin=1in]{geometry}

\setlist[itemize]{itemsep=1pt, topsep=2pt}
\setlist[enumerate]{itemsep=1pt, topsep=2pt}
\hypersetup{
  pdftitle={Strongly Polynomial Parallel Maximum Flow Revisited},
  pdfauthor={Adam Karczmarz and Paweł Pilarski},
  pdfkeywords={maximum flow, parallel algorithm, work-depth tradeoff, strongly polynomial}
}

\title{Strongly Polynomial Parallel Maximum Flow Revisited\thanks{Supported by the National Science Centre (NCN) grant no. 2022/47/D/ST6/02184.}}
\newcommand{\email}[1]{\href{mailto:#1}{#1}}
\author{Adam Karczmarz\thanks{University of Warsaw, Poland. \email{a.karczmarz@mimuw.edu.pl}.} \and
Paweł Pilarski\thanks{University of Warsaw, Poland. \email{pawel.piilarski@gmail.com}.}}
\date{}

\theoremstyle{plain}
\newtheorem{theorem}{Theorem}[section]
\newtheorem{lemma}[theorem]{Lemma}
\newtheorem{corollary}[theorem]{Corollary}
\newtheorem{definition}[theorem]{Definition}
\Crefname{algocf}{Algorithm}{Algorithms}

\newcommand{\reverse}{\overleftarrow}
\newcommand{\eps}{\varepsilon}
\newcommand{\RR}{\mathbb{R}}
\newcommand{\Ot}{\widetilde{O}}
\newcommand{\poly}{\operatorname{poly}}
\newcommand{\ex}{\mathrm{ex}}
\newcommand{\val}{\mathrm{val}}
\newcommand{\Eess}{E_{\mathrm{ess}}}
\newcommand{\Ep}{E_{\mathcal{P}}}
\newcommand{\changed}[1]{\textcolor{red}{#1}}
\newcommand{\ExactFlow}{\mathsf{ExactFlow}}
\newcommand{\cost}{\mathrm{cost}}

\newcommand{\codeStyle}[1]{\mathsf{#1}}
\newcommand{\dsStyle}[1]{\mathsf{#1}}
\newcommand{\approxFlow}{\codeStyle{ApproxFlow}}
\newcommand{\RemoveRoot}{\codeStyle{RemoveRoot}}
\newcommand{\PostProcess}{\codeStyle{PostProcess}}
\newcommand{\CreateCompact}{\codeStyle{Create\-Aux\-Instance}}
\newcommand{\Initialize}{\codeStyle{Initialize}}
\newcommand{\ConvertBasic}{\codeStyle{ConvertBasic}}
\newcommand{\MaxCap}{\codeStyle{MaxCap}}
\newcommand{\SendFlow}{\codeStyle{SendFlow}}
\newcommand{\RerouteS}{\codeStyle{RerouteShortcut}}
\newcommand{\Fe}{F_{\mathrm{e}}}
\newcommand{\Fs}{F_{\mathrm{s}}}
\newcommand{\ab}{\Gamma}
\newcommand{\roots}{\mathcal{R}}
\newcommand{\Vess}{\roots_{\mathrm{ess}}}
\newcommand{\CP}{\mathcal{P}}
\newcommand{\ACP}{E_{\CP}}
\newcommand{\Eaux}{E_{\mathrm{ess}}}
\newcommand{\rev}[1]{\overleftarrow{#1}}
\newcommand{\cu}{\bar{u}}
\newcommand{\cG}{\bar{G}}
\newcommand{\gV}{V}
\newcommand{\gEi}{E^\circ}
\newcommand{\gGi}{G^\circ}
\newcommand{\ui}{u^{\circ}}
\newcommand{\gEStr}{E_{\mathrm{gap}}}
\newcommand{\gEA}{E_{\mathrm{abd}}}
\newcommand{\gEF}{E_{\mathrm{free}}}
\newcommand{\gEP}{\ACP}
\newcommand{\nr}{r}
\newcommand{\res}[1]{u^{#1}}
\newcommand{\supp}{\mathrm{supp}}
\newcommand{\gEbd}{E_{\mathrm{bound}}}
\newcommand{\defeq}{\stackrel{\scriptscriptstyle{\mathrm{def}}}{=}}
\newcommand{\Fone}{F_{\mathrm{one}}}
\newcommand{\Ftwo}{F_{\mathrm{two}}}
\newcommand{\dsTcInit}{\dsStyle{TcInit}}
\newcommand{\dsTcAdd}{\dsStyle{TcAdd}}
\newcommand{\dsTcWit}{\dsStyle{TcWitList}}
\newcommand{\dsTcCover}{\dsStyle{TcCover}}

\newcommand{\WitRoute}{\dsStyle{WitRoute}}

\tikzstyle{tikzfig}=[baseline=-0.25em,scale=0.5]
\pgfdeclarelayer{edgelayer}
\pgfdeclarelayer{nodelayer}
\pgfsetlayers{edgelayer,nodelayer,main}
\tikzstyle{none}=[inner sep=0mm]
\tikzstyle{default diedge}=[->]
\tikzstyle{peretz}=[fill=white,draw=black,shape=circle,text width=0.6cm,align=center]

\begin{document}
\maketitle

\begin{abstract}
We study the maximum flow problem in directed networks with real capacities in the parallel setting. For a network with $n$ vertices and $m$ arcs, we show that a randomized parallel implementation of a variant of the strongly polynomial max-flow algorithm of Dadush, Orlin, Sidford, and Végh~\cite{DBLP:conf/soda/DadushOSV26} runs in
$\Ot(mn)$ work and $\Ot(m)$ depth.
This improves upon the previously described tradeoffs between work and depth for strongly polynomial parallel maximum flow algorithms: earlier $\Ot(n^3)$-work algorithms have $\Ot(n^2)$ depth~\cite{DBLP:journals/jacm/GoldbergT88,DBLP:journals/jal/ShiloachV82a}, while the known $\Ot(m)$-depth approach uses $\Ot(mn^3)$ work~\cite{fastermincostflow}.
\end{abstract}

\section{Introduction}
Maximum flow is one of the most fundamental and well-studied polynomial-time problems on graphs.
Given a graph $G=(V,E)$ with $n$ vertices and $m$ arcs with non-negative capacities, a source $s\in V$, and a sink $t\in V$, the problem asks to send as much flow from $s$ to $t$ as possible while maintaining flow conservation at every $v\in V\setminus\{s,t\}$ and not exceeding arc capacities.

Numerous max-flow algorithms have been described since the problem was first shown to be polynomial-time solvable~\cite{DBLP:journals/jacm/EdmondsK72,dinic1970algorithm}.
These early algorithms~\cite{DBLP:journals/jacm/EdmondsK72,dinic1970algorithm} were \emph{strongly polynomial}, meaning that (1) their running time, measured in the number of elementary arithmetic operations, depends only on the graph size parameters $n$ and $m$, and (2) they use polynomial space on a Turing machine under the standard binary encoding of capacities.
In contrast, the running time of \emph{weakly polynomial} methods may also depend polynomially on the encoding length of the capacities.
Typically, weakly polynomial max-flow algorithms are studied for integer capacities, and their efficiency is analyzed as a function of $n$, $m$, and $\log U$, where $U$ denotes the maximum capacity in the instance.

Weakly polynomial methods are currently the fastest for integer capacities of subexponential magnitude.
The first max-flow algorithm running in $\Ot((nm)^{1-\delta}\log U)$ time for some $\delta>0$ appeared in the late 1990s~\cite{DBLP:journals/jacm/GoldbergR98}.
Recent years have brought striking improvements, including $m^{1+o(1)}\log U$-time algorithms~\cite{DBLP:journals/jacm/ChenKLPGS25,DBLP:conf/focs/Brand0PKLGSS23,DBLP:conf/focs/Brand0KLMGS24}, as well as even faster methods tailored to denser graphs~\cite{DBLP:conf/stoc/BrandLLSS0W21,DBLP:conf/focs/BernsteinBLST25}.
Notably, the interior-point-based developments~\cite{DBLP:journals/jacm/ChenKLPGS25,DBLP:conf/focs/Brand0PKLGSS23,DBLP:conf/focs/Brand0KLMGS24,DBLP:conf/stoc/BrandLLSS0W21} also solve the more general min-cost flow problem within similar bounds.

At the same time, there has been no polynomial improvement in strongly polynomial max-flow running time, as a function of $n$ and $m$ alone, since the 1980s, when $\Ot(nm)$-time algorithms were obtained~\cite{DBLP:journals/jcss/GalilN80,DBLP:journals/jcss/SleatorT83} by accelerating the blocking-flow method~\cite{dinic1970algorithm} with dynamic trees.
More recent work on strongly polynomial max-flow~\cite{orlin2013max,DBLP:conf/soda/DadushOSV26} focused on matching the $O(nm)$ bound, understanding the barrier behind it, and exploiting additional structure.
In particular,~\cite{orlin2013max,DBLP:conf/soda/DadushOSV26} give faster strongly polynomial algorithms for dense instances with a limited number of \emph{capacitated} (that is, finite-capacity) arcs.
If $m_c=O(m)$ denotes the number of vertices plus the number of capacitated arcs, then one algorithm in~\cite{DBLP:conf/soda/DadushOSV26} runs in $\Ot(n^{\omega-1}m_c)$ time, where $\omega<2.372$ is the matrix multiplication exponent~\cite{DBLP:conf/soda/AlmanDWXXZ25}.

\paragraph{Parallel setting.}
Our focus is on parallel shared-memory algorithms for max flow.
We use the standard \emph{work-depth} viewpoint (e.g.,~\cite{blelloch1996programming,jaja1992introduction,DBLP:journals/jal/ShiloachV82a}) when measuring the efficiency of a parallel algorithm.
The \emph{work} is the total number of operations performed, and the \emph{depth} is the length of the longest chain of sequential dependencies.
Since we care about polynomial speedups only, we state all work and depth bounds in $\Ot(\cdot)$ notation.
At that level of granularity, there is no need to distinguish between standard shared-memory models such as EREW and CREW, because they are equivalent up to polylogarithmic overheads~\cite{jaja1992introduction}.

Parallel computation of max flow has also been studied extensively.
The known algorithms, however, generally exhibit a trade-off between work and depth: lower depth is obtained at the price of higher work.
This applies to both weakly and strongly polynomial algorithms.
The recent almost-linear-time max-flow algorithms~\cite{DBLP:journals/jacm/ChenKLPGS25,DBLP:conf/focs/Brand0PKLGSS23,DBLP:conf/focs/Brand0KLMGS24} are highly sequential.
On the weakly polynomial side, a recent parallel min-cost flow algorithm of~\cite{parallel_int_flow} yields improved bounds for max flow on sufficiently dense instances, achieving near-linear work there together with $\Ot(\sqrt{n}\poly(\log U))$ depth.
Older randomized parallel results also give very low depth in capacity-dependent regimes, but only with high polynomial work and with running times that still depend explicitly on $\log U$ or on related encoding parameters~\cite{DBLP:journals/combinatorica/KarpUW86,DBLP:journals/orl/OrlinS93}.

Strongly polynomial parallel max-flow received most attention in the 1980s.
Blocking-flow-based methods~\cite{dinic1970algorithm} generally perform $\Theta(n)$ blocking-flow augmentations in sequence.
The $\Ot(m)$-time blocking-flow algorithm for acyclic networks~\cite{DBLP:journals/jcss/SleatorT83} seems difficult to parallelize in a non-trivial way; however, one can compute a blocking flow in $\Ot(n^2)$ work and $\Ot(n)$ depth~\cite{DBLP:journals/jal/ShiloachV82a,DBLP:journals/jal/Vishkin92}, which yields $\Ot(n^3)$ work and $\Ot(n^2)$ depth overall~\cite{DBLP:journals/jal/ShiloachV82a}.
A comparable work-depth pair is also achieved by the push-relabel framework~\cite{DBLP:journals/jacm/GoldbergT88}.
A near-linear depth bound can be extracted from Orlin's strongly polynomial min-cost flow framework~\cite{fastermincostflow}: the algorithm is presented as a sequence of $\Ot(m_c)$ shortest-path computations, so plugging in a polylog-depth parallel shortest-path routine yields $\Ot(m_c)$ depth at the cost of $\Ot(m_cn^3)$ work.
Parallel max-flow was also investigated in~\cite{DBLP:journals/jpdc/PeretzF22}\footnote{Peretz and Fischler~\cite{DBLP:journals/jpdc/PeretzF22} describe a relatively simple algorithm and argue that it can be implemented in $\Ot(n^4)$ work and $\Ot(n)$ depth. However, we were unable to verify the claimed $O(n)$-iteration bound as stated. We discuss this issue, and give a counterexample to the argument from~\cite{DBLP:journals/jpdc/PeretzF22}, in~\Cref{sec:error}.}.

There is strong evidence that strongly polynomial $\Ot(1)$ depth may be out of reach.
For binary-encoded capacities, max-flow is known to be $\mathrm{P}$-complete under log-space reductions, and thus also under $\mathrm{NC}$ reductions~\cite{DBLP:journals/tcs/GoldschlagerSS82}.
The relevant reduction to a Circuit Value Problem (see also~\cite{limits_book}) uses a sparse network with capacities of order $2^{\Theta(n)}$.
Since arithmetic on polynomial-bit integers has polylogarithmic depth (see, e.g.,~\cite{DBLP:journals/combinatorica/KarpUW86}), a strongly polynomial $\Ot(1)$-depth max-flow algorithm would imply algorithms even for such exponential-capacity networks, proving $\mathrm{P}=\mathrm{NC}$.
Achieving truly sublinear $O(n^{1-\epsilon})$ depth (where $\epsilon>0$) for strongly polynomial max-flow would also constitute a breakthrough, as no truly sublinear-depth parallel algorithms are known for Circuit Value Problems (cf.~\cite{limits_book}).

To sum up, some of the known strongly polynomial parallel algorithms achieve a natural near-linear depth barrier.
This raises the question whether one can reach near-linear depth without giving up near-state-of-the-art (sequential) strongly polynomial work.

\subsection{Our results}
We answer this question in the affirmative.
Our main result is the following.

\begin{theorem}\label{thm:main}
There is a Monte Carlo randomized strongly polynomial max-flow algorithm with $\Ot(mn)$ work and $\Ot(m_c)$ depth.
The returned result is correct with high probability.
\end{theorem}
By verifying the output and rerunning on failure, we obtain a Las Vegas algorithm with expected $\Ot(mn)$ work and expected $\Ot(m_c)$ depth. Verification checks the output flow's feasibility and then tests if $t$ is unreachable from $s$ in the residual network.
A naive parallel BFS-based reachability test requires $\Ot(m)$ work and $\Ot(n)$ depth, which is dominated by the stated bounds.

In the strongly polynomial regime, the algorithm behind Theorem~\ref{thm:main} matches the best known work and the best known depth simultaneously, up to polylogarithmic factors.
To prove Theorem~\ref{thm:main}, we describe a parallel implementation of a variant of the recent algorithm of Dadush, Orlin, Sidford, and V\'egh~\cite{DBLP:conf/soda/DadushOSV26}.

An incremental transitive closure data structure with $O(nm)$ total update time~\cite{DBLP:journals/tcs/Italiano86} is a crucial element of state-of-the-art strongly polynomial max-flow algorithms for sparse graphs~\cite{orlin2013max,DBLP:conf/soda/DadushOSV26}.
However, that data structure may need $\Omega(n)$ depth to process a single batch of updates.
Along the way, we obtain an auxiliary result that may be of independent interest: a parallel batch-incremental data structure for transitive~closure.

\begin{theorem}\label{thm:inc-tc}
Let $G$ be a directed graph subject to batches of arc insertions.
There is a Monte Carlo randomized data structure that explicitly maintains the reachability matrix of $G$ with $\Ot(mn)$ total work, where $m$ is the number of arcs in the final graph.
Each batch is processed in $\Ot(n^{1/3})$~depth.
The maintained matrix is correct with high probability.
\end{theorem}
Thus, we preserve near $O(nm)$ total work while reducing the depth per batch from $O(n)$ to $\Ot(n^{1/3})$.
To the best of our knowledge,
no prior work gives nontrivial work–depth bounds for batch-incremental transitive closure while retaining near-optimal total work.
While there is a substantial body of work on batch-dynamic parallel graph data structures, e.g.,~\cite{DBLP:conf/spaa/AcarABD19,DBLP:conf/spaa/AndersonBT20,DBLP:conf/spaa/LiuSYDS22}, much of it concerns undirected graphs and polylogarithmic-depth regimes.

\subsection{Roadmap} In~\Cref{sec:preliminaries} we set some notation and introduce the preliminary
notions needed to describe the algorithm we build on. In~\Cref{sec:alg_overview} we give a high-level overview of the algorithm of~\cite{DBLP:conf/soda/DadushOSV26}.
Building on that, in~\Cref{sec:technical_overview} we describe the challenges in parallelizing its framework and sketch how they are addressed in our parallel implementation.
This order lets us present our contribution at a high level before turning to its technical justification.

The following part of the paper provides that justification through the details of our parallelization of~\cite{DBLP:conf/soda/DadushOSV26}.
We do not give a complete exposition of the algorithm of Dadush et al.~\cite{DBLP:conf/soda/DadushOSV26}; instead we focus on the needed adjustments.
In~\Cref{sec:additional} we cover some further details of the algorithm of~\cite{DBLP:conf/soda/DadushOSV26} that are easily parallelizable.
\Cref{sec:changes} is devoted to the adjustments we make to the framework of~\cite{DBLP:conf/soda/DadushOSV26}.
In~\Cref{sec:parallelizing_basics} we deal with the parallel bottlenecks of the adjusted algorithm, except for the batch-incremental parallel incremental transitive closure data structure, which is presented in~\Cref{sec:hitting_set_itco}.
Finally,~\Cref{sec:removing_v_s_and_t_v} addresses parallelizing a certain preconditioning step of~\cite{DBLP:conf/soda/DadushOSV26} (discussed in~\Cref{sec:preliminaries}).

\section{Preliminaries}\label{sec:preliminaries}
We follow the notation of~\cite{itco_flow}.\footnote{In the following, we often refer to the full version~\cite{itco_flow} of~\cite{DBLP:conf/soda/DadushOSV26}.}
Let $G = (V, E)$ be a directed multigraph where $V$ is a set of vertices and $E$ is a set of distinguishable arc objects; in particular, parallel arcs are allowed. Let $n = |V|$ and $m = |E|$. For an arc $e$ with tail $i$ and head $j$, we write $e = (i, j)$; this notation does not identify $e$ with the ordered pair $(i, j)$. Let $u$ be the capacity vector, so that $u_e\in \RR_{\geq 0}\cup \{\infty\}$ represents the capacity of an arc $e\in E$. Let $m_c$ denote the number of arcs in $E$ with positive finite capacities plus $n$.

We will discuss maximum flow instances on different graphs and capacity vectors.
We denote by $(G, u)$ a maximum flow instance where $G=(V,E)$ is the underlying network and $u$ is a capacity vector.
The graph $G$ has a distinguished source $s\in V$ and sink $t\in V$.

For a vector $f\in \RR_{\geq 0}^E$, let the excess of a vertex $v\in V$ be defined as
\[\ex_f(v) = \sum_{e\in E:\,\operatorname{head}(e)=v}f_e - \sum_{e\in E:\,\operatorname{tail}(e)=v}f_e.\]
$f$ is a \emph{feasible flow} or \emph{flow} if for every arc $e$, it satisfies $0 \le f_e \le u_e$ and for every vertex $v \in V \setminus\{s, t\}$ it satisfies $\ex_f(v) = 0$. The \emph{value} of flow $f$ is defined as
\[\val(f) = -\ex_f(s) = \ex_f(t).\]
The value of the maximum flow for an instance $(G, u)$ is denoted by $\nu(G, u)$.

As in~\cite{itco_flow}, we make the following additional assumptions about the input network~$G$:
\begin{itemize}
   \item $n\geq 4$ and $m\geq n-1$. This is easy to guarantee by adding and/or removing isolated vertices and possibly adding arcs to make the resulting graph weakly connected.
   \item $G$ includes distinguished infinite-capacity auxiliary arcs, denoted $(v, s)$ and $(t, v)$, for all ${v\in V\setminus\{s,t\}}$.
Clearly, the auxiliary arcs do not change the maximum flow value.
However, an optimum flow might use these arcs.
In~\Cref{sec:removing_v_s_and_t_v}, we show how to convert any flow into a flow whose value is at least the original value and that does not use the additional arcs, with $\Ot(m)$ work and $\Ot(n)$ depth.
Since $m\geq n-1$, introducing the auxiliary arcs can increase the network size by a constant factor.
\item $E$ contains no arc from $s$ to $t$. Every such direct arc is saturated in any optimum flow and can be removed after accounting for its capacity.
\item There is no $s\to t$ path in $G$ consisting of infinite-capacity arcs only.
If this were the case, the solution is unbounded.
A graph search detects such a path in $\Ot(m)$ work and $\Ot(n)$ depth.
\item
Finally, each $e\in E$ has a distinct reverse arc $\reverse{e}\in E$, where $\reverse{\reverse{e}} = e$ and, if $e = (i, j)$, then $\reverse{e} = (j, i)$. Oppositely directed input arcs may be paired arbitrarily; for every as-yet-unpaired input arc, we add a distinct reverse arc of capacity $u_{\reverse{e}} = 0$ and pair the two. This can increase the number of arcs by at most a factor of 2.
\end{itemize}

For a flow vector $f$, let $u^f_e$ be the \emph{residual capacity} of arc $e$, defined as
\[u^f_e = u_e - f_e + f_{\reverse{e}}.\]

\subsection{Acyclic basic flows}\label{sec:preliminaries_acyclic_basic}
An important ingredient of the algorithm presented in \cite{itco_flow} is the notion of an \emph{acyclic basic~flow}.

\begin{definition}[Acyclic flow]
   A feasible flow $f$ is acyclic if the set of arcs $e$ such that $f_e > 0$ forms an acyclic graph.
\end{definition}
A basic feasible flow is a basic feasible solution to a max-flow linear program. The following definition contains two equivalent conditions characterizing a basic flow without referring to the theory of linear programming.
\begin{definition}[Basic flow]\label{def:basic_flow}
   A feasible flow $f$ is called \emph{basic} if one of the following equivalent definitions holds:
   \begin{enumerate}
      \item\label{basic_forest} The multiset of \emph{undirected} arcs $\{u,v\}$ obtained from the arcs $e=(u, v)$ satisfying $0 < f_e < u_e$ forms a forest where $s$ and $t$ are in different trees.
      \item\label{basic_convex} $f$ is not a non-trivial convex combination of some two feasible flows.
   \end{enumerate}
\end{definition}

Note that in item~\ref{basic_forest} of~\Cref{def:basic_flow}, a flow cannot be basic if there is an arc $e$ such that $0 < f_e < u_e$ and $0 < f_{\reverse{e}} < u_{\reverse{e}}$: the undirected arcs obtained from $e$ and $\reverse{e}$ would form a cycle and thus the considered multiset would not be a forest.

In \cite{itco_flow}, a basic flow is defined using item~\ref{basic_convex} of~\Cref{def:basic_flow} and it is proven in~\cite[Lemma~2.2]{itco_flow} that for each basic flow, item~\ref{basic_forest} holds as well.
For completeness, we prove the other implication in~\Cref{lem:basic_flow} in~\Cref{sec:omitted}.

\section{Overview of the algorithm of Dadush, Orlin, Sidford, and Végh}\label{sec:alg_overview}
In this section we give a high-level sketch of the $\Ot(nm)$-time algorithm from \cite{itco_flow} with a special focus on the ingredients that need to be taken care of in the parallel implementation.
An informal skeleton of the framework of~\cite{itco_flow} outlined in this section can be found in~\Cref{alg:framework-skeleton}.
There, the parallel bottlenecks of the framework have been marked red.
Further details of the algorithm from~\cite{itco_flow} are covered in~\Cref{sec:additional}.

The algorithm of~\cite{itco_flow} uses a~\emph{scaling} approach, i.e., it operates in iterations that gradually improve the approximation error of the constructed solution. Specifically, it maintains a feasible flow $f$ along with an estimate $\eps$ that upper-bounds the flow left to be sent in the residual network $(G, u^f)$.
In each iteration the error~$\eps$ either becomes at most $\Gamma^{-3}$ times its previous value, where $\Gamma := 4n^2$, or vanishes completely. This is achieved as follows.

The arcs of $(G, u^f)$ are divided into categories.
First, \emph{abundant} arcs are arcs with such large capacity that they can be treated like infinite-capacity arcs. More precisely, $e$ is called abundant if $u_e^f \ge \Gamma\eps$. If an arc becomes abundant at some point, it stays abundant in subsequent iterations.
Note that, by the definition of $\eps$, $G$ does not contain an $s\to t$ path composed of abundant arcs.
If both $e$ and $\reverse{e}$ are abundant, $e$ is called \emph{free}.

Abundant and free arcs define a partition of $(G,u^f)$ into \emph{free components}.
The free component containing $s$ is a set of vertices $S$ reachable from $s$ using paths composed of abundant arcs only.
Symmetrically, the free component of $t$ consists of vertices $T$ that can reach $t$ via abundant arcs.
Recall that $S\cap T=\emptyset$.
Finally, the free components of $V\setminus (S\cup T)$ are the same as the connected components of the subgraph of free arcs induced on $V\setminus (S\cup T)$.

Define $\Ep\subseteq E$ to be the arcs that connect distinct free components.
The arcs outside~$\Ep$ and all but $O(n)$ abundant arcs can be effectively ignored, as each free component $Y$ is ``abundantly'' strongly connected, that is, enough flow can be routed freely within $G[Y]$ via~$O(|Y|)$ abundant arcs only.
Now, the non-abundant arcs in $\Ep$ are partitioned into

\begin{enumerate}
   \item \emph{Essential} arcs $\Eess$ that satisfy $\Gamma\eps > u_e^f \ge \Gamma^{-5}\eps$.
   \item \emph{Small} arcs that satisfy $u_e^f < \Gamma^{-5}\eps$.
\end{enumerate}

Roughly speaking, in a single iteration, the flow $f$ is augmented by a $\Gamma^{-5}$-approximate maximum flow in an \emph{auxiliary graph} $\bar{G}$ obtained from the residual network $(G, u^f)$ by
\begin{enumerate}[(1)]
   \item ignoring arcs $E\setminus \Ep$ and the small arcs,
   \item preserving the essential arcs, and
   \item compressing connections between $s$, $t$, and the endpoints of $\Eess$ via abundant paths.
\end{enumerate}
In~\cite{itco_flow}, it is proven that, in this way, both the number of iterations and the total number of essential arcs over all iterations are $O(m_c)$.

For correctness and efficiency, the computed approximate maximum flow needs to satisfy additional requirements. One uses an \emph{approximate flow solver} $\approxFlow$ defined as follows.

\begin{definition}\label{def:approx_flow}An approximate flow solver $\approxFlow(G, u, M)$ is a procedure that, for an instance $(G, u)$ with $m:=|E(G)|$ and a precision parameter $M\in\RR_{>0}$, computes an acyclic basic flow $f$ and $e^*\in E(G)\cup\{\bot\}$, where we put $u^f_\bot:=0$, such that
   \[ \nu(G, u^f) \le m\,u^f_{e^*} \le m\nu(G, u^f) \le \frac{\nu(G, u)}{M}. \]
\end{definition}

\begin{algorithm}[t!]
    \caption{\label{alg:framework-skeleton}The high-level skeleton of the max-flow algorithm of Dadush et al.~\cite{itco_flow}. The main parallel bottlenecks -- the iteration count, $\approxFlow$, and updating the transitive closure -- are highlighted in red.}
   Maintain a feasible flow $f$, an error bound $\epsilon$, and an incremental set of abundant arcs\;
   Maintain the free components of $(G,u^f)$ and the subset $\Ep$\;
   Set up an incremental transitive closure data structure $\mathcal{D}$ for the abundant subgraph\;
   {\color{red}\While{$\epsilon>0$}{
      {\color{black}
      Identify the current iteration's essential arcs $\Eess\subseteq \Ep$\ and their endpoints $X_i$\;
      Compress the reachability via abundant arcs between $X_i\cup\{s,t\}$ into $H$ using $\mathcal{D}$\;
      Build an auxiliary instance $(\bar{G},\bar{u})$ from essential arcs and $H$\;
      }
      $f'\gets \approxFlow(\bar{G},\bar{u},\poly(n))$\;
      {\color{black}
      Augment $f$ by $f'$ (potentially introducing auxiliary arcs to $G$)\;
      Decrease $\epsilon$ and identify new abundant arcs and free components\;
      }
      Insert the new abundant arcs into $\mathcal{D}$\;
   }}
   Reroute flow through auxiliary arcs in $f$ to original arcs of $G$\;
   \Return{$f$}\;
\end{algorithm}

Above, producing an arc $e^*$ whose residual capacity approximates the error serves as a mechanism to help guarantee low bit complexity of the values computed by the algorithm.

As shown in~\cite{itco_flow}, $\approxFlow$ can be implemented by solving a certain \emph{exact} max-flow instance with integral capacities bounded by $\poly(n,M)$ (obtained from $(G,u)$ by scaling, rounding, and capping) followed by an $\Ot(m)$-time strongly polynomial step converting the flow into a basic acyclic one~\cite{DBLP:journals/jcss/SleatorT83,basic_flow_algorithm_paper}.
If one uses the state-of-the-art weakly polynomial algorithm~\cite{DBLP:journals/jacm/ChenKLPGS25} for the former, $\approxFlow$ runs in $m^{1+o(1)}$ time.

Let us now discuss in more detail how abundant subgraph compression is performed in the auxiliary network $\bar{G}$ to guarantee the algorithm's efficiency.
It is useful to see why the compression is necessary at all.
In principle, in every iteration one could simply include all the abundant and essential arcs in $\bar{G}$. However, in each of $\Theta(m_c)$ iterations we could possibly have just a single essential arc and $\Theta(m)$ abundant arcs, which would make the total cost of $\approxFlow$ calls $\Omega(m_cm)$ even if a linear-time implementation of $\approxFlow$ were used.

In~\cite{itco_flow}, the problem of abundant arc compression is reduced to the following \emph{transitive cover} problem.
Given a graph $G=(V,E)$ and a subset $S\subseteq V$, a graph $H=(V',E')$ is a transitive cover of $S$ with respect to $G$ if $S\subseteq V'\subseteq V$, $E'$ is a subset of the transitive closure of $G$ induced on~$V'$, and pairwise reachability in $G$ between vertices in~$S$ is preserved in $H$.
With that notion in hand, in iteration $i$ with $m_i$ essential arcs $\Eess$ whose endpoint set is $X_i$, it would be enough (roughly speaking) if the auxiliary graph $\bar{G}_i$ contained $\Eess$, plus a transitive cover $H_i$ of $X_i\cup\{s,t\}$ with respect to the abundant subgraph $G^*$ of $G$.\footnote{Actually, for technical reasons, the auxiliary network in~\cite{itco_flow} is constructed in a slightly more sophisticated manner; see~\Cref{sec:additional}.}

To implement this approach, one needs an efficient way to compute possibly small~transitive covers of subsets $X_i\cup \{s,t\}$.
Note that the sets $X_i$ are revealed online, and recall that the abundant subgraph changes over time: it is subject to $O(m)$ arc~insertions throughout.
Dadush et al.~\cite{itco_flow} capture what the algorithm needs here using the following \emph{incremental transitive cover} (ITCO) data structure problem.
\begin{definition}\label{def:itco}
An ITCO data structure supports the following~operations\footnote{In~\cite{itco_flow}, the ITCO data structure also supports the $\dsTcWit$ operation returning a so-called \emph{witness list}. This object is used to  eventually reroute the obtained flow back to original edges of $G$. We present an alternative way of rerouting the flow, avoiding the explicit use of witness lists, in~\Cref{sec:witroute}.}:
\begin{itemize}
   \item $\dsTcInit(V)$: Initialize with a graph $G^*=(V,E)$, where $E$ is an empty set.
   \item $\dsTcAdd(E'\subseteq V'\times V')$: Add arcs $E'$ to $E$.
   \item $\dsTcCover(S\subseteq V)$: Return a transitive cover $(V',E')$ of $S$ wrt. $G^*$.
\end{itemize}
\end{definition}

Dadush et al.~\cite{itco_flow} implement the ITCO data structure as follows.
The graph $G^*$ is maintained in an \emph{incremental transitive closure} data structure~\cite{DBLP:journals/tcs/Italiano86} that supports arc insertions and $O(1)$-time reachability queries, provides access to the paths certifying the pairwise reachability, and has $O(nm)$ total update time.
$\dsTcAdd$ simply passes insertions to that data structure.
$\dsTcCover(S)$, on the other hand, returns either the entire graph $G^*$ or a graph $(S,F)$, where $F$ contains one arc $uv$ for each pair $u,v\in S$ such that there exists a path $u\to v$ in $G^*$ --- \emph{whichever is smaller}.
Note that in the latter case, the graph has size $O(|S|^2)$ and can be obtained in $O(|S|^2)$ time using reachability queries.
Observe that if $m^*$ arcs are added to $G^*$ eventually and transitive cover queries are issued for sets of sizes $s_1,\ldots,s_k$ with $\sum_{i=1}^k s_i=s^*$,
then the total size of transitive covers produced is at most $\sum_{i=1}^k \min(s_i^2,m^*)$.
By summing separately for $s_i<\sqrt{m^*}$ and $s_i\geq \sqrt{m^*}$, we can bound the sum by $O(s^*\sqrt{m^*})$.
With the ITCO data structure, the total size of auxiliary networks $\bar{G}_i$ produced over all iterations is $O(m_c\sqrt{m})$.
This allows bounding the total time spent on $\approxFlow$ calls by $(m_c\sqrt{m})^{1+o(1)}$, which is $O(nm)$ for $m=O(n^{1.99})$.\footnote{For dense graphs $m=\Theta(n^2)$, using the dense solver~\cite{DBLP:conf/stoc/BrandLLSS0W21} in place of~\cite{DBLP:journals/jacm/ChenKLPGS25} yields $\Ot(nm)$ time.}

\section{High-level overview of the parallel implementation}\label{sec:technical_overview}
In this section we discuss the main challenges in parallelizing the algorithm of~\cite{DBLP:conf/soda/DadushOSV26} outlined in~\Cref{sec:alg_overview} and sketch how these challenges are addressed in our algorithm.

\subsection{Challenges in parallel implementation}
The framework of~\cite{DBLP:conf/soda/DadushOSV26} does not parallelize well out of the box and would have $\Ot(nm)$ depth if implemented literally.
The first obstacle is the number of iterations of the main loop (see~\Cref{alg:framework-skeleton}): the framework may perform $\Theta(m_c)$ of them, as each iteration processes a set of essential arcs, and $O(m_c)$ essential arcs arise in total~\cite{DBLP:conf/soda/DadushOSV26}.
Since the number of iterations is a lower bound on the depth of any parallel implementation, obtaining $\Ot(m_c)$ depth without reducing the iteration count would require each iteration to run in $\Ot(1)$ depth.

There are two main parallel bottlenecks inside an iteration.
The first one is $\approxFlow$.
An implementation of $\approxFlow$ using the almost-linear weakly polynomial max-flow algorithm of~\cite{DBLP:journals/jacm/ChenKLPGS25} for the rounded instance would have depth linear in the number of arcs, which is too much because the total size of all auxiliary instances in~\cite{DBLP:conf/soda/DadushOSV26} is $O(m_c\sqrt{m})$.
While there are lower-depth weakly polynomial approaches~\cite{parallel_int_flow,DBLP:journals/combinatorica/KarpUW86,DBLP:journals/orl/OrlinS93}, they come with higher work.
Fortunately, the recent weakly polynomial parallel algorithm of~\cite{parallel_int_flow} provides a good enough work-depth trade-off to fit in our desired work/depth budgets.
There is another problem, though: the known efficient methods for enforcing the additional structural requirements of~$f'$ -- acyclicity~\cite{DBLP:journals/jcss/SleatorT83} and basicness~\cite{basic_flow_algorithm_paper} -- are inherently sequential.

The second bottleneck is incremental transitive closure.
Updating the transitive closure of a dense graph in $\Ot(1)$ depth per inserted arc is easy if one is willing to spend, say, $\Ot(n^2)$ work per update.
The difficulty is to achieve the $\Ot(nm)$ total work bound needed here, or $\Ot(n)$ amortized work per inserted arc.
The algorithm of~\cite{DBLP:journals/tcs/Italiano86} and other natural approaches detect new reachability relations by extending already-discovered paths one arc at a time.
Such an approach can easily use $\Theta(n^2)$ depth when processing $O(n)$ insertions in sequence.

\subsection{Adjustments to the framework of~\cite{DBLP:conf/soda/DadushOSV26}}
Our first modification is to reduce the number of iterations.
Recall that the original main loop may perform $\Theta(m_c)$ iterations, each of which triggers one batch update to the transitive closure structure.
Designing a structure with $\Ot(nm)$ total work and $\Ot(1)$ depth per batch proved challenging, and our incremental transitive closure data structure (\Cref{thm:inc-tc}) only gives $\Ot(n^{1/3})$ depth per batch without increasing the total work.
To compensate for that, we incorporate an idea already present in Orlin's framework~\cite{orlin2013max}: we ensure that iterations process a larger set of essential arcs at once so that the total number of iterations is reduced to $O(m_c/n^{1/3})$.
The price is that the residual capacity range used to define the essential arcs in an iteration has to become wider.
We show that this does not affect correctness.

However, the widened range is no longer guaranteed to be polynomially bounded in $n$ as in~\cite{DBLP:conf/soda/DadushOSV26}, so it is unclear whether using a weakly polynomial max-flow algorithm such as \cite{DBLP:journals/jacm/ChenKLPGS25,parallel_int_flow} could still be efficient enough inside the $\approxFlow$ subroutine (see~\Cref{def:approx_flow}).
In those iterations we therefore fall back to a low-depth (but high-work) exact strongly polynomial algorithm~\cite{fastermincostflow}.
To prevent the high $\Ot(m_cn^3)$ work of Orlin's parallel algorithm~\cite{fastermincostflow} from becoming the sole work bottleneck, we ensure that whenever we use the strongly polynomial subroutine~\cite{fastermincostflow}, the auxiliary network has $O(n^{1/3})$ vertices in the worst case and $O(n^{1/3})$ capacitated arcs on average.

\subsection{Making an approximate flow acyclic and basic}
The near-linear procedures~\cite{basic_flow_algorithm_paper,DBLP:journals/jcss/SleatorT83} used in~\cite{DBLP:conf/soda/DadushOSV26} to turn an approximate max-flow obtained via scaling, rounding, and capping in the first step of~$\approxFlow$ into an acyclic one and then into a basic one seem difficult to parallelize.
Within the original framework they create an~$\Ot(m_c\sqrt{m})$ depth bottleneck.
We therefore avoid them altogether.
Instead, we enforce the required structure already on the rounded instances by means of an isolation argument for flows~\cite{uniqueness_lemma_paper}.
By sampling polynomially bounded positive integral arc costs, one can guarantee with high probability that the resulting minimum-cost flow routing a prescribed demand vector is unique.
A unique minimum-cost flow is automatically acyclic and basic.
Thus, the computed approximate flows satisfy the structural requirements for the scaled, rounded, and capped capacities.

The framework still requires these properties with respect to the original residual capacities.
Acyclicity is preserved when we translate the rounded solution back to the original instance, but basicness may be lost.
We show that the resulting flow is nevertheless close enough to basic that it can be repaired in depth linear in the number of capacitated arcs of the auxiliary network.
The total number of such arcs over the entire algorithm is $O(m_c)$.

\subsection{Incremental transitive closure with improved depth per batch insertion}
To prove Theorem~\ref{thm:inc-tc}, we combine repeated matrix squaring with the shortcut-edges technique of Bernstein~\cite{DBLP:journals/siamcomp/Bernstein16}.
Related combinations of these ideas have previously proved useful in static parallel APSP computation~\cite{DBLP:conf/soda/KarczmarzS21} and in sequential approximate incremental APSP data structures~\cite{DBLP:conf/esa/KarczmarzL19}.

Let $H\subseteq V$ be a random subset of size $\Theta(n^{2/3}\log n)$, called the \emph{hitting set}.
A standard fact~\cite{DBLP:journals/siamcomp/UllmanY91} implies that, with high probability, $H$ contains a vertex of every path on $n^{1/3}$ vertices in any fixed collection of $\poly(n)$ paths.
As long as $H$ is chosen independently~of the updates, this remains true for paths coming from any intermediate version of the evolving~graph.

Our data structure maintains four interlocking components:
\begin{enumerate}[(1)]
\item Short paths: for every $v\in V$, it keeps a subset $S_v$ of vertices currently reachable from $v$ that includes all vertices within distance $n^{1/3}$ from $v$.
\item $H\times H$ reachability: for pairs of hitting-set vertices, it maintains reachability in a dense auxiliary graph $G'=(H,E')$, where $uv\in E'$ iff $v\in S_u$.
\item $V\times H$ reachability: equivalently, it maintains reachability from $H$ to all of $V$ in the reverse graph by adding shortcut arcs to the underlying incremental graph.
\item $V\times V$ reachability: finally, it extends the same idea to all source vertices in $V$.
\end{enumerate}
The first, third, and fourth components are implemented using a collection of parallel BFS runs resumed after each batch insertion and truncated after performing $\Ot(n^{1/3})$ steps; this guarantees $\Ot(n^{1/3})$ depth per batch. The dense graph $G'$ on $H$ can be updated in $\Ot(1)$ depth per batch with $\Ot(|H|^3)=\Ot(n^2)$ total work by using a deterministic repeated-squaring-based incremental transitive closure data structure.
Altogether, this yields $\Ot(mn)$ total work and~$\Ot(n^{1/3})$ depth per insertion batch.

A final issue is adaptivity.
When the data structure is plugged into the max-flow algorithm of~\cite{DBLP:conf/soda/DadushOSV26}, its~randomness remains independent of the updates it receives, at least~until the data structure errs, in which case we simply count this as a failure of the overall Monte Carlo algorithm.
This is because the max-flow algorithm only queries the maintained reachability relation during its main loop, and that relation is uniquely determined by the~updates.

\section{Additional details of the algorithm}\label{sec:additional}
In this section, we review some more details of the algorithm of~\cite{itco_flow} that have been omitted in the high-level overview in~\Cref{sec:alg_overview}.
We will need that to argue that these details do not pose a challenge from the parallelization standpoint.

We start with a few more details regarding the free components. Each free component has some chosen root and is internally stored as a rooted reachability tree. For the free components of $s$ and~$t$, the roots must always be $s$ and $t$, respectively. The set of roots is denoted by~$\mathcal{R}$.

We now describe how the auxiliary network is constructed for a given set of essential arcs $\Eess$.
\begin{enumerate}
   \item The algorithm computes the set of \emph{essential roots} $\Vess$, which contains the roots of the free components corresponding to the endpoints of the essential arcs.
   \item $\Vess$ is passed as input to the $\dsTcCover$ query of the incremental transitive cover data structure, which returns a set of vertices $V' \supseteq \Vess$ and a set of arcs $E'$.
   \item The set of vertices of the resulting auxiliary network $\bar{G}$ consists of $V'$ and the endpoints of~$\Eess$.
   \item The set of arcs of the resulting auxiliary graph consists of $\Eess$, $E'$ and free arcs connecting the endpoints of essential arcs to their respective roots. Arcs from $E'$ have infinite capacities. For arcs from $\Eess$, instead of their original capacities, the algorithm uses a so-called \emph{safe capacity vector}~$r$, which is maintained throughout the algorithm.
\end{enumerate}
The safe capacity vector, used and maintained for technical reasons (see~\cite{itco_flow}), is defined as a set of capacities $r$ such that $r_e \le u^f_e$ for all $e \in E$, but the value of the maximum flow for instances $(G, r)$ and $(G, u^f)$ is the same. The algorithm starts with $r := u^f$.

The computation of a flow approximation (via $\approxFlow$) on the auxiliary network is followed by a post-processing step called $\PostProcess$, necessary to maintain the algorithm's invariants.

To explain the most important non-trivial invariant, we first need to define \emph{boundary arcs} as arcs $e = (i, j)$ such that $\{i, j\} \cap \{s, t\} \neq \emptyset$, that is, one of the two endpoints is $s$ or $t$. The invariant says that before and after every iteration, each arc $e$ must satisfy one of the following.
\begin{itemize}
   \item $e$ is a boundary arc.
   \item $u^f_e = 0$ or $u^f_{\reverse{e}} = 0$.
   \item $e$ is a so-called \emph{gap arc}.
\end{itemize}
To ensure this invariant, the procedure $\PostProcess$ modifies the flow on the \emph{non-boundary} arcs of the auxiliary network in the following way:
\begin{itemize}
   \item If an arc $e = (i, j)$ has a very small residual capacity in one of the directions, then, assuming w.l.o.g. that $u^f_e$ is small, $\PostProcess$ modifies the flow $f_e$ and $f_{\reverse{e}}$ so that $u^f_e = 0$ is satisfied. This change makes the flow excess positive for vertex $j$ and negative for vertex~$i$. To fix this, the positive excess is rerouted by passing as much flow as possible using the arc $(j, t)$,
   and the remaining flow is routed via the infinite-capacity arc $(j, s)$. To fix the negative excess, we send as much flow as possible to $i$ using the arc $(s, i)$, and the remaining flow is sent via the infinite-capacity arc $(t, i)$. This rerouting might decrease the value of the flow by $u^f_e$, which is small.
   \item If an arc has too much residual capacity in both directions, then the algorithm modifies the safe capacity vector in such a way that the arc becomes a \emph{gap arc}, informally meaning that its residual capacity in each direction is slightly larger than its safe capacity.
\end{itemize}

The main motivation for maintaining the invariant by post-processing is to limit the number of essential arcs throughout the algorithm. For arcs such that $u^f_e = 0$, we have $u^f_{\reverse{e}} = u_e + u_{\reverse{e}}$. Because the value of $u_e + u_{\reverse{e}}$ remains constant throughout the algorithm, and because $\eps$ in each iteration is multiplied by a factor of at most $\Gamma^{-3}$, these arcs can be in the essential arc range $\Gamma^{-5}\eps \le u^f_{\reverse{e}} = u_e + u_{\reverse{e}} < \Gamma\eps$ at most a constant number of times. Gap arcs, on the other hand, have a useful property that they turn into free arcs in the auxiliary network after one iteration, meaning that each gap arc can only be essential once. \cite{itco_flow} proves a slightly stronger fact that across all the iterations there are in total $O(m_c)$ gap arcs and essential arcs. It can be shown for boundary arcs that, among them, there can be at most $O(n)$ essential arcs in total.

After post-processing, the algorithm needs to translate the flow from the auxiliary graph onto the original one. Translating the flow on arcs from $\Eess$ is trivial. If an arc from $E'$ carries positive flow, we add it to the graph $G$ as an \emph{extension arc}. Dadush et al.~\cite{itco_flow} prove, using the properties of basic flows, free components and post-processing, that there are at most $2n$ extension arcs introduced throughout.

To translate the flow from the arcs between the root of the free component and the endpoint of an essential arc,~\cite{itco_flow} use \emph{shortcut arcs}. Shortcut arcs are free arcs from every vertex of the component to its root. When merging two components, neither of which is the component of $s$ or~$t$, the smaller component is merged into a bigger one. The root of the bigger component becomes the root of the combined component. For each vertex in the smaller component we add a new shortcut arc to account for the root change. When merging a component of $s$ or $t$ with any other component, the latter is always merged into the former. This way of combining components guarantees that there will be $O(n \log n)$ shortcut arcs.

To eventually reroute the flow on the shortcut arcs, the algorithm first converts the flow from shortcut arcs in the component into required excesses for each vertex. The algorithm then, for each component, eliminates the leaves of the component's tree one by one by deciding how much flow to pass using the removed arc to make the excess of the removed vertex 0.

To reroute the extension arcs, the algorithm gradually replaces their flow by a flow on corresponding paths of abundant ares that had existed in the graph when a corresponding extension arc was created.
We discuss this in more detail in~\Cref{sec:witroute}.

A subtlety in this algorithm is worth noting.  Let the set of extension arcs be denoted by $\Fe$ and let the set of shortcut arcs be denoted by $\Fs$. The arcs of the graph used throughout the algorithm are defined as a disjoint sum $E \uplus \Fe \uplus \Fs$. This can create multi-arcs as long as the multi-arcs come from different sets of this sum; by the convention from~\Cref{sec:preliminaries}, they remain distinct arc objects. Each of these sets individually does not contain multi-arcs.

\section{Adjustments to the algorithm of Dadush, Orlin, Sidford, and Végh}\label{sec:changes}

In this section we describe the adjustments we make to the algorithm~\cite{itco_flow} outlined in Sections~\ref{sec:alg_overview}~and~\ref{sec:additional}.
\Cref{alg:max-flow} contains the pseudocode of that algorithm, reproduced verbatim (see \cite[Algorithm~1]{itco_flow}), with our changes marked.
Note that we have not defined much of the notation used in the pseudocode, since we only focus on the changes needed by our parallelization.
Refer to~\cite{itco_flow} for the missing notation.

\subsection{Definition changes}\label{sec:definition_changes}
Our main change to the framework~\cite{itco_flow}
is a more general definition of an essential arc.
\begin{definition}[cf. {\cite[Definition~4.7]{itco_flow}}]\label{def:essential_arc}
   Let $f$ be an $\eps$-optimal flow in $(G, u)$ for $\eps > 0$ and let $k\in \RR\cup \{\infty\}$ be defined as
   \[k = \min\{k \in \RR : k \ge 3, |\{e\in E_\mathcal{P}:\Gamma^{-(k + 2)}\eps \le u_e^f < \Gamma\eps\text{ or }\Gamma^{-(k + 2)}\eps \le u_{\reverse{e}}^f < \Gamma\eps \}| \ge 2n^{1/3}\}.\]
   We say that $k$ is an \emph{essential arc exponent}.
   We say that $e \in E$ is an \emph{essential arc} with respect to~$k$ if $e \in E_{\mathcal{P}}$
   and
\[\Gamma^{-(k+2)}\eps < u^f_e < \Gamma \eps \quad \text{or} \quad \Gamma^{-(k+2)}\eps < u^f_{\reverse{e}} < \Gamma \eps.\]
   Additionally, if there are fewer than $2n^{1/3}$ essential arcs for $k < \infty$, we arbitrarily select some of the arcs that satisfy
\[u^f_e = \Gamma^{-(k+2)}\eps \quad \text{or} \quad u^f_{\reverse{e}} = \Gamma^{-(k+2)}\eps.\]
   We consider those arcs and their reversals essential and select enough of them so that the number of essential arcs is exactly $2\lceil n^{1/3}\rceil$.
\end{definition}

The formula for $k$ in~\Cref{def:essential_arc} guarantees that for $k < \infty$ there are at least $\Omega(n^{1/3})$ essential arcs in a single iteration. If $k>3$, then $\Gamma^{-(k+2)}\eps$ is the residual capacity of some arc, and if $k=\infty$, then this threshold is 0. It also allows us to prove there are $\Ot(m_c / n^{1/3})$ iterations.

This alteration also influences the definition of small arcs so that they satisfy $u^f_e \le \Gamma^{-(k + 2)}\eps$ and are non-essential. The flow missed by ignoring the small arcs can still be bounded by $n^2\Gamma^{-(k+2)}\eps$.

From the definition of the essential arcs we immediately get the following lemma.
\begin{lemma}\label{lem:essential_implies_interval}
   Every essential arc $e$ satisfies \[\Gamma^{-(k + 2)}\eps \le u_e^f < \Gamma\eps\quad\text{or}\quad\Gamma^{-(k + 2)}\eps \le u_{\rev{e}}^f < \Gamma\eps.\]
\end{lemma}
\begin{proof}
   In at least one direction, an essential arc either satisfies $\Gamma^{-(k + 2)}\eps < u_e^f < \Gamma\eps$ or is selected from the set of arcs satisfying $u_e^f = \Gamma^{-(k + 2)}\eps$. Combining the two cases gives us the inequality from the lemma.
\end{proof}

Naturally, we also need to adjust item (A) in the definition of valid successors~\cite[Definition~4.9]{itco_flow} from $\eps' \le \Gamma^{-3}\eps$ to $\eps' \le \Gamma^{-k}\eps$.

{
\begin{algorithm}[t!]
    \caption{\label{alg:itco_flow_changed}The pseudocode for the max-flow algorithm of~\cite{itco_flow} with our changes marked red.}\label{alg:max-flow}
    \KwData{An instance $(\gGi,\ui)$ with $\gGi=(\gV,\gEi)$ and $\nu(\gGi,u)<\infty$.}
    \KwResult{A maximum flow $f^\star$ in  $(\gGi,\ui)$.}
    $(G,u)\gets (\gGi,\ui)$ ; $E\gets \gEi$ ;
    $\Fe\gets\emptyset$ ; $\Fs\gets\emptyset$ \;
    $(f,\nr,\eps,\roots)\gets\Initialize$ \tcp*{Obtain initial valid iterate}
    ${\cal T}_i\gets(\{i\},\emptyset)$ for all $i\in V$ \;
   $\dsTcInit(V)$ ; $\dsTcAdd(\gEA)$    \tcp*{Initialize transitive cover data structure}
    \While{$\eps>0$}{
       \changed{
       $\Gamma^{-(k+2)}\eps\gets$ computed according to \Cref{def:essential_arc}\;
       }
       $\Eaux\gets\left\{e \in \gEP\,:\,  \Gamma^{\changed{-(k+2)}}\eps \changed{<} \res{f}_e< \ab\eps \text{ or } \Gamma^{\changed{-(k+2)}}\eps \changed{<} \res{f}_{\rev{e}}< \ab\eps\right\}$\;
       \changed{\lIf{$k<\infty\text{ and }|\Eaux| < 2n^{1/3}$}{$\Eaux \gets \Eaux\cup \Eaux' \cup \rev{\Eaux'}$ where $\Eaux'$ is any subset of arcs such that $|\Eaux \cup \Eaux' \cup \rev{\Eaux}'| =  2\lceil n^{1/3} \rceil, \forall_{e \in \Eaux'} u^f_e = \Gamma^{\changed{-(k+2)}}\eps$}}
   $\Vess\gets$ roots of components with essential arcs incident\;
   \lIf{$\{e\in \gEP\changed{\backslash \Eaux}\,:\,\res{f}_e\changed{\le}\Gamma^{\changed{-(k+2)}}\eps\}\neq\emptyset$}{$\delta_1\gets n^2\max\left\{\res{f}_e: e\in \gEP\changed{\backslash \Eaux}, \res{f}_e\changed{\le}\Gamma^{\changed{-(k+2)}}\eps\right\}$}\lElse{$\delta_1\gets 0$ }
      $(\cG,\cu,H)\gets\CreateCompact(\Vess,\Eaux)$ \;
    \If{$s,t\in\Vess$}{
       \changed{
         \lIf{$\Gamma^{-(k+2)}\eps = \Gamma^{-5}\eps$}{
   $(y,e^\star)\gets \approxFlow(\cG,\cu,\Gamma^{5})$
   $\delta_2\gets |E(\cG)| \cu^y_{e^\star}$
         }
         \lElse{
            $y\gets \ExactFlow(\cG,\cu)$
            $\delta_2\gets 0$
         }
       }
        \lFor{$e\in \Eaux$}{$\SendFlow(e,y_e)$ \tcp*[f]{Add $y$ to flow} \label{line:sendflow}}
    \lFor{$e=(i,j)\in H$}{
        $E\gets E\uplus\{e,\rev{e}\}$;
        $u_e\gets\infty$; $u_{\rev{e}}\gets 0$; $f_e\gets y_e$; $f_{\rev{e}}\gets 0$  \label{line:sendflow2}}
   }
   \lElse{$\delta_2\gets 0$}
       $\delta\gets\delta_1+\delta_2$ \tcp*{current flow is $\delta$-optimal}
    $\PostProcess(\delta,(\Eaux\setminus\gEbd)\uplus H)$ \;
    $E\gets E\setminus\{e,\rev{e}\,:\,e\in H, f_e=0\}$ \tcp*{\hspace{-2mm}Keep only new arcs with positive flow}
    $\Fe\gets \Fe\cup (H\cap\supp(f))$ \;

      $\eps\gets 2n^2\delta$ \;
        Update $\gEA$, $\gEF$, $\CP$, $\gEP$ and $\roots$ \;
    $J\gets$ set of newly abundant arcs ;
    $\dsTcAdd{(J\setminus H)}$   \tcp*{Update transitive cover}
  \lFor{$k$ removed from $\roots$}{$\RemoveRoot(k)$}
        }
    $\RerouteS(\Fs)$ \;
    $h\gets \changed{\WitRoute(f|_{\gEA})}$ \tcp*{Map flow back to original arcs}
     \lFor{$e\in\gEi$}{
        $f^\star_e\gets h_e$ if $e \in \gEA$ and $f^\star_e\gets f_e$ if $e \notin \gEA$}
  \Return{$f^\star$}
\end{algorithm}
}
\subsection{Algorithm changes}
\begin{definition}[Exact flow]
   An \emph{exact flow solver} is a procedure $\ExactFlow(G, u)$ that computes, for an instance $(G, u)$, a maximum flow $f$ that is also acyclic and basic.
\end{definition}

The algorithm changes in the following way. At the beginning of the iteration we check if the value $k$ from~\Cref{def:essential_arc} equals $3$.
If this is the case, we use $\approxFlow$ to compute the approximate flow.
Otherwise,
there are $O(n^{1/3})$ essential arcs,
and we use $\ExactFlow$ instead of $\approxFlow$.

\subsection{Additional background}
Before we analyse \Cref{alg:itco_flow_changed}, we will first present the notation and definitions from~\cite{itco_flow} that will be used in the lemmas, lemma changes and proof changes.

Most sets of arcs are defined based on some parameters. $E_{\text{abd}}(f, \eps)$ and $E_{\text{free}}(f, \eps)$ denote respectively the sets of abundant arcs and free arcs for the flow $f$ and for the given value of~$\eps$. The set of arcs between different components $\Ep$ is defined in \cite{itco_flow} as $\Ep(f, \eps)$. The need for $f$ and $\eps$ comes down to the fact that the free components depend only on the free arcs. The set of gap arcs is defined as
\[E_{\text{gap}}(f, r, \eps) \defeq \{e \in \Ep(f, \eps) : u^f_e \ge r_e + 2\Gamma^{-2}\eps\ \text{and}\ u^f_{\reverse{e}} \ge r_{\reverse{e}} + 2\Gamma^{-2}\eps\}.\]
The set of boundary arcs is defined as
\[E_{\text{bound}} \defeq \{(s, i), (i, s), (t, i), (i, t) : i \in \mathcal{R} \backslash \{s, t\}\}.\]

We will recall a definition of valid iterates from \cite[Definition~4.6]{itco_flow}.
\begin{definition}[{\cite[Valid Iterate]{itco_flow}}] \label{def:proper}
We say that $(f,\nr,
\varepsilon,\roots)$ is a \emph{valid iterate} in the instance $(G,u)$ if
\begin{enumerate}[label=(\Alph*)]
   \item\label{it:eps-opt} $f$ is an $\varepsilon$-optimal feasible flow in $(G,u)$;
   \item\label{it:safe}  $\nr$ is a safe capacity vector for $f$.
   \item\label{it:bounds} For every arc $e\in \gEP(f,\varepsilon)$, at least one of the following holds:  $\res{f}_e=0$;  $\res{f}_{\rev{e}}=0$; $e\in\gEStr(f,\nr,\varepsilon)$;  or $e\in\gEbd(\roots)$.
   \item\label{it:oneway}  $\res{f}_{si}=0$ or $\res{f}_{it}=0$ for $i\in\roots$.
   \item\label{it:roots} $\roots$ forms a set of roots of the free components, i.e., every component of $\CP_{f,\varepsilon}$ contains exactly one node from $\roots$, and $s,t\in \roots$.
\end{enumerate}
\end{definition}
\Cref{it:bounds} from the above definition will be especially important for us.

Let us recall the complete \cite[Definition 4.9]{itco_flow} with the changes from \Cref{sec:definition_changes} applied. The \emph{admissible extension} $(G, u) \oplus (F_{\text{one}}, F_{\text{two}})$ used in the definition can be interpreted as a graph $(G, u)$ with $F_{\text{one}}$ as the set of newly added extension arcs and $F_{\text{two}}$ as the set of newly added shortcut arcs.
\begin{definition}[{\cite[Valid Successors]{itco_flow} with a change from \Cref{sec:definition_changes}}] \label{def:successor}
Let \linebreak${(G',u')=(G,u)\oplus (\Fone,\Ftwo)}$ be an admissible extension of $(G,u)$.
Let $(f,\nr,\varepsilon,\roots)$ be a valid iterate in $(G,u)$, and let $k$ be its essential arc exponent according to~\Cref{def:essential_arc}. Let $(f',\nr',\varepsilon',\roots')$ be a valid iterate in $(G',u')$. We say that $(f',\nr',\varepsilon',\roots')$ is a \emph{valid successor} of $(f,\nr,\varepsilon,\roots)$ if the following hold:
\begin{enumerate}[label=(\Alph*)]
\item\label{it:eps-decrease}
$\varepsilon'\le \ab^{-k}\varepsilon$.
\item\label{it:arc-prox} For every ${e\in E(G)}$, we have
${f'_e-f_e+f_{\rev{e}}\le \min\{\nr_e,\varepsilon\}+\ab^{-2}\varepsilon<2\varepsilon}$,
and
for every \linebreak ${e\in E(G')\setminus E(G)}$, we have ${f'_e\le 3\varepsilon}$.
\end{enumerate}
\end{definition}

\subsection{Lemma and proof changes}\label{sec:lemmachanges}
Because of our changes to the definitions of essential arcs and valid successors, we need to slightly redefine and/or reprove lemmas that use the essential arc definition. Most of these proofs are nearly identical to the original proofs and we will only describe the differences.

We will start with a lemma that is the main reason why gap arcs are useful.
\begin{lemma}[{cf.~\cite[Lemma~4.10]{itco_flow}}]
   Let $(f, r, \eps, \mathcal{R})$ be a valid iterate in $(G, u)$. Let $(G', u')$ be an admissible extension of $(G, u)$ with a valid iterate $(f', r', \eps', \mathcal{R}')$ that is a valid successor of $(f, r, \eps, \mathcal{R})$.
   \begin{enumerate}[label=(\roman*)]
      \item $E_{\text{abd}}(f, \eps) \subseteq E_{\text{abd}}(f', \eps')$ and $E_{\text{free}}(f, \eps) \subseteq E_{\text{free}}(f', \eps')$. \label{4.10:i}
      \item Every arc $e \in E_{\text{gap}}(f, r, \eps)$ is $(f', \eps')$-free. \label{4.10:ii}
   \end{enumerate}
\end{lemma}
\begin{proof}[Proof change]
   The proof is the same as it relies on the fact that $\eps' \le \Gamma^{-3}\eps$, which is still true because the essential arc exponent $k \ge 3$.
\end{proof}
To prove \cite[Lemma~6.2]{itco_flow}, which shows that the algorithm produces valid successors in each iteration, the authors of~\cite{itco_flow} first prove \cite[Lemma~6.3]{itco_flow}, which needs to include our changes.
\begin{lemma}[{cf.~\cite[Lemma~6.3]{itco_flow}}]
   Given an instance $(G, u)$ with a valid iterate $(f, r, \eps, \mathcal{R})$ and some essential arc exponent $k \ge 3$, let $\mathcal{R}_{\text{ess}}$ be the set of essential vertices with respect to $k$; let $E_{\text{ess}}$ be the set of essential arcs and let
   \[\delta_1 = n^2\max \left(\left\{u^f_e : e \in \gEP\backslash \Eaux,\ u^f_e \le \Gamma^{-(k + 2)}\eps\right\}\cup\{0\}\right).\]
   For the instance $(\bar{G}, \bar{u})$ returned by $\CreateCompact(\mathcal{R}_{\text{ess}}, E_{\text{ess}})$, the following holds: \[\text{val}(\bar{G}, \bar{u}) \ge \text{val}(G, u^f) - \delta_1.\] If $s \notin \mathcal{R}_{\text{ess}}$ or $t \notin \mathcal{R}_{\text{ess}}$, then $\text{val}(\bar{G}, \bar{u}) = 0$.
\end{lemma}
\begin{proof}[Proof change]
   The proof is the same as the original except that, when defining the graph $G'$, (1) we delete all arcs $e \in E \backslash \Eaux$ such that $u_e + u_{\reverse{e}} \le \Gamma^{-(k+2)}\eps$ and (2) we set $u_e'=0$ for all non-essential arcs $e \in E$ such that $\reverse{e} \in E_{\text{abd}}(f, \eps)$ and $u^f_e \le \Gamma^{-(k+2)}\eps$. Apart from that, $u'$ is defined in the same way. Again, the following holds:
   \[\text{val}(G, u^f) \ge \text{val}(G', u') \ge \text{val}(G, u^f) - \delta_1\]
   because capacity decreases on arcs internal to a free component cannot affect a finite cut in $G'$, since every such cut respects the free components. All remaining capacity decreases are on arcs in $\gEP\backslash\Eaux$, and their total is at most $\delta_1$. The rest of the proof stays the same.
\end{proof}

\begin{lemma}[{cf.~\cite[Lemma~6.2]{itco_flow}}]\label{lem:lemma62}
   Assume $(f, r, \eps, \mathcal{R})$ is a valid iterate in $(G, u)$ at the beginning of an iteration of the while loop in Algorithm 1, and let $(f', r', \eps', \mathcal{R}')$ denote the corresponding quantities at the beginning of the next iteration in the instance $(G', u')$. Then, $(f', r', \eps', \mathcal{R}')$ is a valid iterate and $(f', r', \eps', \mathcal{R}')$ is a valid successor of $(f, r, \eps, \mathcal{R})$. Moreover, the following stronger version of \cite[Definition~4.9(B)]{itco_flow} holds:
   \begin{itemize}
      \item For every $e \in E(G), f'_e - f_e + f_{\reverse{e}} \le \min\{r_e, \eps\} + \Gamma^{-3}\eps$, and for every $e \in E(G') \backslash E(G)$, $f'_e \le 2\eps$.
   \end{itemize}
\end{lemma}
\begin{proof}[Proof change]
   The proof is almost the same. When proving property~\cite[Definition~4.9(A)]{itco_flow}, we have
   \[\delta = \delta_1 + \delta_2 \le n^2\Gamma^{-(k + 2)}\eps + \Gamma^{-(k + 2)}\eps \le \Gamma^{-(k+1)}\eps,\]
   \[\eps' = 2n^2\delta \le \Gamma^{-k}\eps,\]
   and the proof of~\cite[Definition~4.9(B')]{itco_flow} is the same because \[\eps' \le \Gamma^{-k}\eps \le \Gamma^{-3}\eps.\]
\end{proof}
\Cref{lem:lemma62} says that $\tau + 1$ is a valid successor of iteration $\tau$. Using this, the authors can prove \cite[Lemma~6.4]{itco_flow}, which says that for any iterations $\tau$ and $\tau'$ such that $\tau < \tau'$, $\tau'$ is a valid successor of $\tau$. Property \cite[Definition~4.9(A)]{itco_flow} is obvious; property \cite[Definition~4.9(B)]{itco_flow} can be proved in the same way since the proof uses only the stronger version of property \cite[Definition~4.9(B)]{itco_flow} from \Cref{lem:lemma62} and the original version of \cite[Definition~4.9(A)]{itco_flow}, which is weaker.

The following lemmas \cite[Lemma~4.8]{itco_flow}, \cite[Lemma~6.5]{itco_flow} and \cite[Lemma~6.7]{itco_flow} are used in \cite[Lemma~6.10]{itco_flow} to show that there are $O(m_c)$ essential arcs. \cite[Lemma~6.5]{itco_flow} states that throughout the algorithm $u_{\reverse{e}}^{f^{(\tau)}} = 0$ for non-boundary arcs $e = (i, j) \in E^\circ$ such that $u_e = \infty$, where $E^\circ$ is the set of arcs from the original instance. The proof of this lemma assumes that extension arcs can create multi-arcs; in particular, there could be an essential extension arc $(i, j)$ which is different from arc $e = (i, j) \in E^\circ$. \cite[Lemma~6.7]{itco_flow} shows that there are at most $2n$ extension arcs and at most $n\log_2n$ shortcut arcs, and its proof does not change. \cite[Lemma~4.8]{itco_flow} does change:
\begin{lemma}[{cf.~\cite[Lemma~4.8]{itco_flow}} with changes]\label{lemma:4.8_with_changes}
   Let $(f, r, \eps, \mathcal{R})$ be a valid iterate in $(G, u)$ and let $k$ be an essential arc exponent. For every essential arc $e$, either $e \in E_{\text{gap}}(f, r, \eps)$, or $e \in E_{\text{bound}}(\mathcal{R})$, or $\Gamma^{-(k+2)}\eps \le u_e + u_{\reverse{e}} < \Gamma\eps$.
\end{lemma}
\begin{proof}
   In~\Cref{it:bounds} of~\Cref{def:proper}, the cases $u^f_e = 0, u^f_{\reverse{e}} = u_e + u_{\reverse{e}}$ and $u^f_e = u_e + u_{\reverse{e}}, u^f_{\reverse{e}} = 0$ are combined with \Cref{lem:essential_implies_interval} to get $\Gamma^{-(k+2)}\eps \le u_e + u_{\reverse{e}} < \Gamma\eps$.
\end{proof}

We are now ready to discuss the key lemma bounding the number of essential arcs.
\begin{restatable}[{cf.~\cite[Lemma~6.10]{itco_flow}}]{lemma}{essentialarcsbound}\label{lem:essential_bound}
   The total number of essential arcs throughout all iterations is $O(m_c)$.
\end{restatable}
\begin{proof}[Proof change]
   By~\cite[Lemma~6.7]{itco_flow}, there are $O(n)$ extension arcs, and every extension arc will turn into a gap arc in the iteration after it is essential and will be a free arc in the following iteration. Shortcut arcs cannot be essential. It remains to bound the number of essential arcs among the arcs from the original instance.

   There are 3 types of essential arcs according to \Cref{lemma:4.8_with_changes}:
   \begin{itemize}
      \item Regular arcs: $\Gamma^{-(k+3)} \eps \le u^f_e + u^f_{\reverse{e}} = u_e + u_{\reverse{e}} < \Gamma \eps$.
      \item Gap arcs.
      \item Boundary arcs.
   \end{itemize}
   Just as in the original proof, we show that each arc can be of each type at most a constant number of times.

   An arc can be regular in at most 3 iterations.
   Let $\eps^{(i)}$ and $k^{(i)}$ be the values of $\eps$ and $k$, respectively, in iteration $i$. For every iteration $i$ we have $\eps^{(i + 1)} \le \Gamma^{-k^{(i)}}\eps^{(i)}$. Let $\tau$ be the first iteration such that the arc $e$ is regular. If iteration $\tau+3$ does not exist, the claim follows immediately. Otherwise, we have
   \[ u_e + u_{\reverse{e}} \ge \Gamma^{-(k^{(\tau)} + 3)}\eps^{(\tau)} \ge \Gamma^{-3 + k^{(\tau + 1)} + k^{(\tau + 2)}}\eps^{(\tau + 3)} \ge \Gamma^3\eps^{(\tau + 3)}. \]
   Thus, the arc is no longer regular at the start of iteration $\tau + 3$. Only two-way finite arcs can be regular, which bounds the number of regular arcs by $O(m_c)$.

   The proof that there are $O(m_c)$ gap arcs is the same.

   Finally, the proof for the boundary arcs is almost the same except that \[\alpha \defeq u^f_e \ge \Gamma^{-(k^{(\tau)}+2)}\eps \ge \Gamma^{4}\eps'\] since \[\eps' \le \Gamma^{-k^{(\tau)} - k^{(\tau+1)} - k^{(\tau+2)}}\eps \le \Gamma^{-(k^{(\tau)} + 6)}\eps.\] We can bound the number of boundary arcs by $O(n)$.
\end{proof}

In~\cite{itco_flow}, the algorithm's correctness (i.e., that the algorithm computes a maximum flow) is shown in~\cite[Lemma~6.9]{itco_flow}.
Its proof is unchanged except for its use of the helper~\cite[Lemma~6.8]{itco_flow} to justify \textsf{WitRoute}; our replacement is justified in~\Cref{sec:witroute}.

The last important lemma shows that the algorithm is strongly polynomial.
\begin{lemma}[{cf.~\cite[Lemma~6.13]{itco_flow}}]
   The algorithm is strongly polynomial. For rational input, all values of $f$, $\eps$, and $r$ remain polynomially bounded in the input encoding length.
\end{lemma}
\begin{proof}[Proof change]
   The value $\Gamma^{-(k+2)}\eps$ is the current residual capacity of some arc or $0$ when $k > 3$, and is $O(\log n)$ bits longer than the previous epsilon when $k = 3$.

   Apart from that, the proof does not change: note that in every iteration for the auxiliary network $(\bar{G}, \bar{u})$ we still return an acyclic basic flow (possibly an exact one) and for $\delta = \delta_1 + \delta_2$ we still have that $\delta_2$ is equal to some residual capacity in the auxiliary network $\bar{G}$ times $|E(\bar{G})|$ and $\delta_1$ is equal to some residual capacity in the network times $n^2$.
\end{proof}

\subsection{Analysis of changes}

\begin{lemma}\label{lem:iteration_progress}
   The number of essential arcs in each iteration, except the last one, is at least $2n^{1/3}$.
\end{lemma}
\begin{proof}
   This follows immediately from \Cref{def:essential_arc}. If $k = \infty$ then at the end of the iteration we have $\eps = 0$, meaning it is the last iteration.
\end{proof}

\begin{corollary}\label{cor:iterations}
   \Cref{alg:max-flow} performs $O(m_c / n^{1/3})$ iterations.
\end{corollary}
\begin{proof}
   By \Cref{lem:iteration_progress}, in every iteration, except the last one, we process at least $2n^{1/3}$ essential arcs, and by \Cref{lem:essential_bound}, there are in total $O(m_c)$ essential arcs throughout all iterations. This bounds the number of iterations by $O(m_c / n^{1/3})$.
\end{proof}

\section{Parallelizing Important Parts}\label{sec:parallelizing_basics}
In this section we describe how to parallelize the important
parts of~\Cref{alg:max-flow}, i.e., the algorithm of~\cite{itco_flow} outlined in~\Cref{sec:alg_overview}, with the changes described in~\Cref{sec:changes}. We omit the implementation of the ITCO data structure, which is deferred to~\Cref{sec:hitting_set_itco}.

\subsection{$\Initialize$ and $\PostProcess$}
For each pair of non-free arcs $e, \reverse{e}$, the algorithm sets $f_e = 0, f_{\reverse{e}} = u_{\reverse{e}}$ or $f_e = u_e, f_{\reverse{e}} = 0$ depending on those arcs.
This can change the flow excess $\ex_f(v)$ to a non-zero value for some vertices $v \notin \{s, t\}$, which requires rerouting the flow using the boundary arcs of those vertices. Because modified boundary arcs are distinct across all $v$, this rerouting can be done in $\Ot(1)$ depth.

$\PostProcess$ is very similar: first, independently for each reverse pair, we process the at most two orientations in an arbitrary fixed order, and then we have to fix the balance, which can be done in parallel per vertex.
\subsection{Maintaining abundant arcs, free arcs and free components}
At initialization, we can check in parallel whether each arc is abundant or free and add it to the corresponding set or sets.
This step takes $\Ot(1)$ depth and $\Ot(m)$ work. To construct the initial free components of $s$ and $t$, we can run a parallel BFS twice,
with each run taking $\Ot(m)$ work and $\Ot(n)$ depth. We then iterate through the rest of the vertices, and if a vertex has no component assigned, we use a parallel BFS on free arcs in order to compute its initial free component. Each BFS has depth $\Ot(\#\text{vertices assigned to a component})$, which in total gives us depth $\Ot(n)$.

We can also store non-abundant arcs in a parallel balanced search tree (parallel BST; see, e.g.,~\cite{BlellochFS16,GilMV91,PaulVW83}) ordered by their residual capacity throughout the algorithm, which allows us to detect new abundant arcs and free arcs in $\Ot(1)$ depth per iteration. Every time we detect a new abundant arc, we check if it connects the free component of $s$ or $t$ to some new component.
\begin{lemma}\label{lem:abundant_bound}
   The number of initially non-abundant arcs that become abundant is $O(m_c)$.
\end{lemma}
\begin{proof}
   By~\cite[Lemma~6.5]{itco_flow}, if an original non-boundary arc $e$ satisfies $u_e=\infty$, then $u^{f}_{\reverse{e}}=0$ throughout the algorithm. Thus, the zero-capacity reverse arcs of such arcs never become abundant. The remaining arcs that may become abundant consist of $O(m_c)$ original arcs whose pair has positive finite capacity, $O(n)$ boundary arcs, and $O(n)$ extension arcs by~\cite[Lemma~6.7]{itco_flow}. Shortcut arcs are free, and hence abundant, when they are introduced. Since $m_c\ge n$, the claimed bound follows.
\end{proof}
By \Cref{lem:abundant_bound}, at most $O(m_c)$ initially non-abundant arcs become abundant throughout the algorithm. For every resulting component merge, we use $\textsf{RemoveRoot}$, which takes $\Ot(n)$ total work.
\subsection{Computing $k$, \texorpdfstring{$\delta_1$}{delta1} and sampling essential arcs}
Let $q=\lceil n^{1/3}\rceil$. Recall that non-abundant arcs are stored in a parallel BST ordered by residual capacity. For this computation, we restrict the BST to arcs in $E_{\mathcal P}$ and regard $e$ and $\rev e$ as one reverse pair. The key of a pair is the largest residual capacity below $\Gamma\eps$ among its two orientations.

If there are fewer than $q$ pairs with a positive key, we set $k=\infty$ and $\Gamma^{-(k+2)}\eps=0$. Otherwise, let $\beta$ be the $q$-th largest positive pair key and set
\[\Gamma^{-(k+2)}\eps=\min\{\Gamma^{-5}\eps,\beta\}.\]
This determines $k\ge 3$. We include in $\Eaux$ every reverse pair whose key is strictly larger than this threshold. If $k<\infty$ and this gives fewer than $2q$ arcs, we add arbitrary reverse pairs whose key equals the threshold until $|\Eaux|=2q$. This is exactly the selection rule in~\Cref{def:essential_arc}.

To compute $\delta_1$, we temporarily remove the pairs in $\Eaux$ from the BST and query the largest remaining residual capacity not exceeding $\Gamma^{-(k+2)}\eps$, taking zero if none exists; $\delta_1$ is $n^2$ times this value. We then reinsert the removed pairs.

The order-statistic, range-reporting, deletion, and reinsertion operations take $\Ot(m_i)$ work and depth in iteration $i$, where $m_i=|\Eaux|$. Since $\sum_i m_i=O(m_c)$ by~\Cref{lem:essential_bound}, their total work and depth are $\Ot(m_c)$.

\subsection{Exact Flow}\label{sec:batch_algorithm}
\begin{lemma}[Parallel Exact Flow]\label{lemma:exact_flow}
   There exists a Monte Carlo randomized parallel implementation of $\ExactFlow$ with $\Ot(n^3m_c)$ work and $\Ot(m_c)$ depth. The algorithm's output is correct with high probability.
   More precisely, for any supplied integer $N\geq 2$ and a fixed constant $\gamma\geq 1$ such that $2mN^\gamma$ fits in $O(1)$ machine words, its failure probability can be made at most $N^{-\gamma}$.
\end{lemma}
This section is devoted to proving~\Cref{lemma:exact_flow}. To that end, we require a low-depth strongly polynomial min-cost flow algorithm.
We refer to the following result of Orlin~\cite{fastermincostflow}.

\begin{lemma}[\cite{fastermincostflow}]\label{lemma:parallel_mincost}
   There exists a parallel algorithm that computes a minimum-cost circulation using $\Ot(m_c\cdot n^3)$ work and $\Ot(m_c)$ depth in a graph $G=(V,E)$ whose $m$ arcs have capacities in $\RR_{\geq 0}\cup \{\infty\}$ and costs in $\RR$, where $m_c$ denotes $n$ plus the number of finite-capacity arcs.
\end{lemma}
\begin{proof}[Proof sketch]
The implications of Orlin's state-of-the-art strongly polynomial min-cost-flow algorithm~\cite{fastermincostflow} for parallel computation of min-cost flow are discussed in the introduction~of~\cite{fastermincostflow}.

We use the enhanced capacity scaling variant of Orlin's algorithm presented in~\cite[Section~10.7]{network_flows_book}.
On an \emph{uncapacitated} min-cost flow problem (with arbitrary real vertex demands and infinite arc capacities, i.e., $m_c=n$), it performs $\Ot(n)$ single-source shortest path (SSSP) computations.
Apart from $\Ot(n)$ SSSP (and single-source-reachability) computations, the algorithm has $\Ot(nm)$ work and $\Ot(n)$ depth.
SSSP, in turn, can be computed in parallel with $\Ot(n^3)$ work and $\Ot(1)$ depth by repeated matrix squaring using min-plus product.

There is a well-known linear-time $\Ot(1)$-depth reduction of min-cost flow with $k$ finite-capacity arcs to an uncapacitated instance with $O(n+k)$ vertices and $O(m)$ arcs (see, e.g.,~\cite[Section~2.4]{network_flows_book}), so that each of the $O(k)$ introduced auxiliary vertices has degree $2$.
In such graphs, SSSP can be computed as efficiently as if there were $O(n)$ vertices and $O(m)$ arcs~\cite[Section~5]{fastermincostflow}.
As a result, in a capacitated instance, each SSSP computation requires $\Ot(n^3)$ work and $\Ot(1)$ depth.
The desired bounds follow by setting $m_c:=n+k$.
\end{proof}

Moreover, we rely on the following isolation lemma for flows, proved in~\cite{uniqueness_lemma_paper}.
\begin{restatable}[{\cite[Theorem~8.1]{uniqueness_lemma_paper}}]{lemma}{uniquenesslemma}\label{lemma:uniqueness}
Let $\overline{\mathcal{MCF}}$ be a feasible instance of the min-cost flow problem with underlying graph $G = (V, E)$, demand vector $b\in \RR^V$, and capacity vector $u\in\RR^E_{\geq 0}$. Let the cost vector $\overline{c}_{e\in E}$ be generated by picking each $\overline{c}_e$ independently and uniformly from $\{1, 2, \cdots, C\}$. Then the probability that $\overline{\mathcal{MCF}}$ does not have a unique solution is at most~$\frac{2m}{C}$.
\end{restatable}
\Cref{lemma:uniqueness} is proven for $C = 4m$ and finite capacities in~\cite[Theorem~8.1]{uniqueness_lemma_paper}.

The proof of~\cite[Theorem~8.1]{uniqueness_lemma_paper} assumes the capacities are finite. However, positive arc costs imply any minimum-cost flow is acyclic, which means that we can decrease the infinite capacities to $\sum_v |b_v|$ without changing the set of optimal flows.

To argue that the proof can also be generalized to any integer value of $C$, we now outline the structure of that proof. For every arc $e\in E$, one defines an event $D(e)$ which can happen for at most 2 cost values $\overline{c}_e$, or, in other words, with probability $\frac{2}{C}$. Then one proves that if for every arc $e$ an event $D(e)$ does not happen, then there is only one optimal solution. Finally, by the union bound, the probability that at least one $D(e)$ occurs -- and hence that uniqueness is not guaranteed -- is at most $\frac{2m}{C}$.

We are now ready to prove~\Cref{lemma:exact_flow}. To implement $\ExactFlow$, i.e., compute an acyclic basic exact maximum flow,
we perform two min-cost flow computations using~\Cref{lemma:parallel_mincost}:
\begin{enumerate}
\item First, we compute the maximum flow value $\nu(G,u)$ only. Note that one can reduce maximum flow on $G$ to minimum-cost circulation on $G$ with an auxiliary infinite-capacity arc from $t$ to $s$ added. The auxiliary $ts$ arc is given cost $-1$, whereas original arcs are all assigned cost $0$.
Note that in such a network, the minimum-cost circulation maximizes flow through $ts$, or, equivalently, it maximizes $s,t$-flow via original arcs.

\item Given the flow value $\nu(G,u)$, we find an acyclic basic $s,t$-flow of that value as follows. For a supplied integer $N\geq 2$ and a fixed integer constant $\gamma\geq 1$ such that $2mN^\gamma$ fits in $O(1)$ machine words, we apply~\Cref{lemma:uniqueness} with $C=2mN^\gamma$ to $G$ with $b(s)=-\nu(G,u)$, $b(t)=\nu(G,u)$ and $b(v)=0$ for $v\in V\setminus\{s,t\}$.
Then, we find a minimum-cost flow $f$ with costs $\overline{c}_e$ picked randomly from~$\{1,\ldots,2mN^\gamma\}$. By~\Cref{lemma:uniqueness}, the failure probability is at most $N^{-\gamma}$.

We now argue that $f$ is an acyclic basic flow.
Indeed, if the obtained $f$ were not a basic flow, by~\Cref{def:basic_flow} it would be a convex combination of two feasible flows $f_1$ and $f_2$ of value $\nu(G,u)$. But then $\cost(f)$ would also be a convex combination of $\cost(f_1)$ and $\cost(f_2)$. Since $f$ is a unique feasible flow with optimal cost, we have $\cost(f) < \cost(f_1)$ and $\cost(f) < \cost(f_2)$.
But for some $\alpha\in(0,1)$, $\cost(f) = (1 - \alpha)\cost(f_1) + \alpha\cdot \cost(f_2)>\cost(f)$, which is a contradiction.
$f$ is also acyclic, since any minimum-cost flow in a graph with exclusively positive arc costs has to be acyclic.
\end{enumerate}

\subsection{ApproxFlow}\label{sec:parallel_approx_flow}
\begin{lemma}[Parallel ApproxFlow]\label{lemma:parallel_approxflow}
   There exists a Monte Carlo randomized algorithm that computes \textsf{ApproxFlow}$(G, u, M)$ in $\Ot(m + n^{1.5} + m_c'n)$ work and $\Ot(\sqrt{n} + m_c')$ depth for $M = \poly(n)$, where $m_c'$ is the number of finite-capacity arcs. The algorithm's output is correct with high probability.
\end{lemma}
This lemma introduces and uses the quantity $m_c'$ instead of $m_c$, because we have defined $m_c = n + m_c'$ and $n$ might be significantly larger than $m_c'$. In order to construct the algorithm proving~\Cref{lemma:parallel_approxflow}, we need the following result of~\cite{parallel_int_flow}:
\begin{lemma}[\cite{parallel_int_flow}]\label{lemma:weakly_polynomial_flow}
   There exists a randomized weakly polynomial algorithm that solves the \emph{minimum-cost} maximum $s,t$-flow problem with polynomially bounded integer capacities and costs with $\Ot(m + n^{1.5})$ work and $\Ot(\sqrt{n})$ depth.
   The algorithm's output is correct with high~probability.
\end{lemma}
Our algorithm is very similar to the construction from \cite{itco_flow} until \textsf{ConvertBasic} is called. The algorithm requires a procedure $\MaxCap$ which, for an instance $(G, u)$, computes an arc $\bar{e}$ such that $u_{\bar{e}} \le \nu(G, u) \le mu_{\bar{e}}$, estimating the flow left in the network. This procedure is used twice, to estimate the flow in the input instance and, at the end, to compute an arc $e^*$ which estimates the leftover flow in the residual network of the returned flow.

$\MaxCap$ finds a path from $s$ to $t$ with the largest minimum capacity and returns any bottleneck arc.
If there is no positive-capacity $s\to t$ path, we let $\MaxCap$ return the dummy arc $\bot$ with capacity $0$. If the first call to $\MaxCap$ returns $\bot$, $\approxFlow$ returns the zero flow and $\bot$ immediately, without performing the scaling step. The same convention is used for the final call to $\MaxCap$ on the residual network.
This can be achieved sequentially in $\Ot(m)$ time using, e.g., Dijkstra's algorithm.
However, such an implementation has $\Theta(n)$ depth.
An alternative way is to sort the arcs~\cite{parallel_sort} by their capacity in descending order and binary search for the earliest arc $e$ such that there is a path from $s$ to $t$ in the graph with arcs later than $e$ in the order removed.
Note that we can use \Cref{lemma:weakly_polynomial_flow} to check whether such an $s\to t$ path exists, since $s\to t$ reachability is a special case of max-flow with unit capacities. Binary search adds a factor of $\Ot(1)$ which still results in $\Ot(m + n^{1.5})$ work and $\Ot(\sqrt{n})$ depth.

The $\approxFlow$ implementation in~\cite{itco_flow} computes a feasible flow $f$ that is approximately maximum by running the weakly polynomial algorithm of~\cite{DBLP:journals/jacm/ChenKLPGS25} on an instance with $\poly(n,M)$-bounded integer capacities obtained by scaling the original capacities by a value $\delta=u_{\bar{e}}/(m^2M)$ (based on the flow estimate found using $\MaxCap$), rounding them down, and capping them at $m^3M$.
After the integer flow is computed, it is scaled back to use the original (unbounded) range of real capacities.

In order to turn an approximate feasible flow into an acyclic basic flow, the algorithm in~\cite{itco_flow} uses an $\Ot(m)$-time procedure $\ConvertBasic$ based on~\cite{basic_flow_algorithm_paper,DBLP:journals/jcss/SleatorT83}, which is difficult to parallelize.
To circumvent that problem, we first guarantee that $f$ is an acyclic basic integer flow as in \Cref{sec:batch_algorithm}.
That is, instead of computing any max-flow, we assign random $\poly(m)$-bounded costs to all the arcs (using~\Cref{lemma:uniqueness}) and find a (unique) minimum-cost maximum flow.
Recall from~\Cref{sec:batch_algorithm} that the obtained minimum-cost flow is acyclic and basic.

After the integer flow $f$ is scaled back as in~\cite{itco_flow}, it will remain acyclic but might unfortunately no longer be basic.
To see this, consider the characterization of a basic flow from item~\ref{basic_forest}~of~\Cref{def:basic_flow}.
After scaling back, some of the ``full'' arcs that satisfied $f_e=u'_e$ with respect to \emph{scaled, rounded, capped, and scaled back} capacities $u'_e=\delta\cdot \min\left(\lfloor u_e/\delta\rfloor, m^3M\right)$ might no longer satisfy $f_e=u_e$.
Instead, $f_e < u_e$ may hold.
Consequently, the set of undirected arcs such that $0 < f_e < u_e$ may not constitute a forest, even though the arcs satisfying $0<f_e<u_e'$
did form a forest.

To deal with that, let us make a slight change to the integer instance. Instead of capping arcs at $m^3M$, which is a bound on the flow of the instance without capping, we cap them at $m^3M + 1$. This guarantees that for arcs such that $u_e = \infty$ we have $u'_e > \val(G, u')$, and, because the returned flow $f$ is acyclic, we also have $f_e\leq \val(G, u')$.

As a result, before scaling back, arcs such that $u_e = \infty$ satisfy $f_e < u'_e$. This implies that if an infinite-capacity arc is not in the forest with respect to capacities $u'$ and is in the forest with respect to capacities $u$, then the only possibility is that $f_e = 0$.
However, scaling back cannot make a zero flow non-zero, and thus $e$ does not satisfy $0 < f_e < u_e$. This proves that when switching from $u'$ to $u$, the set of arcs satisfying $0 < f_e < u_e$ cannot gain any infinite-capacity arcs; it can only gain $O(m'_c)$ finite-capacity arcs.

Our algorithm to convert the flow to a basic one is very similar to the algorithm from \cite{basic_flow_algorithm_paper}, but we will include this observation. This algorithm starts with an empty graph and adds arcs one by one. If after adding a new arc we still have a forest, we do nothing.
However, if the new arc creates a cycle, one can send some flow along the cycle in either direction so that at least one cycle arc satisfies $f_e\in \{0,u_e\}$. This guarantees that at the end of each iteration, the set of arcs that satisfy $0 < f_e < u_e$ forms a forest again.
To guarantee that the flow is basic, one also needs to check whether there is a path from $s$ to $t$ in the obtained forest, and if there is one, one can pass as much flow as possible from $s$ to $t$ so that at least one of the arcs on this path will no longer satisfy $0 < f_e < u_e$, disconnecting $s$ and $t$.
Observe that the above process cannot introduce new flow cycles, as no flow value turns from $0$ to a positive value.
Hence, the obtained flow remains acyclic.

Because of our observation, we can start from a forest made from infinite-capacity arcs satisfying $0 < f_e < u_e$ and add finite-capacity arcs satisfying $0 < f_e < u_e$ one by one. The original algorithm~\cite{basic_flow_algorithm_paper} uses a dynamic tree~\cite{DBLP:journals/jcss/SleatorT83} which allows processing an arc in $\Ot(1)$ work and depth, but we would need to show that, given a forest, we can construct in $\Ot(n)$ work and $\Ot(1)$ depth a dynamic forest data structure supporting path updates in $\Ot(1)$ time.  While such a data structure is plausible, in our case we can also proceed naively, avoiding dynamic trees completely.

\begin{lemma}
Adding an arc to the forest of non-basic arcs requires $\Ot(n)$ work and $\Ot(1)$~depth.
\end{lemma}

\begin{proof}
Between the iterations, we maintain a set of (undirected counterparts of) arcs in the forest.
When an iteration starts, the forest is preprocessed in $\Ot(n)$ work and $\Ot(1)$ depth using the Euler tour technique~\cite{TarjanV85}
so that each tree becomes (arbitrarily) rooted, and every vertex stores its parent and its depth in the tree.
Afterwards, the folklore \emph{pointer jumping} preprocessing (see, e.g.,~\cite{DBLP:journals/tcs/BenderF04}) is applied. That is, for each
$k = 0, \ldots, \lfloor \log n \rfloor$ and each vertex $v$, the ancestor of $v$ that is $2^k$ steps above $v$ is computed inductively.
This also takes $\Ot(n)$ work and $\Ot(1)$ depth.
The preprocessing allows finding, in $\Ot(1)$ time, (1) the root of the tree containing a query vertex $v$, (2) the lowest common ancestor of any two query vertices $u,v$ in the same tree, and (3) the $k$-th ancestor of a query vertex $v$ in its tree for any integer $k\geq 0$.

Now, processing a finite-capacity arc $(a,b)$ is implemented as follows.
We can check in $\Ot(1)$ time whether $a$ and $b$ belong to the same tree by checking if the roots of their respective trees are the same. If not, we add the arc to the forest and finish the iteration.

If $a$ and $b$ are in the same tree, we find the $a$-to-$b$ path in that tree by computing the lowest common ancestor $v$ of these vertices in $\Ot(1)$ time, and then obtaining in parallel the parent arcs of the ancestors of $a$ and $b$ that are closer to $a$ and $b$, respectively, than $v$ is.
This takes $\Ot(n)$ work and $\Ot(1)$ depth.

Having the set of arcs on the path, we can easily compute in parallel how much and in which direction to pass the flow and which arcs will no longer satisfy $0 < f_e < u_e$ and thus have to be removed from the forest. Using this approach, we can also check at the end if there is a path from $s$ to $t$, get all the arcs on the path and pass as much flow as possible.
\end{proof}

Since we spend $\Ot(n)$ work and $\Ot(1)$ depth per processed arc, converting the flow to a basic flow takes $\Ot(nm_c')$ work and $\Ot(m_c')$ depth in total.
This finishes the proof of \Cref{lemma:parallel_approxflow}.

\subsection{\textsf{WitRoute} and \textsf{RerouteShortcut}}\label{sec:witroute}
In \cite{itco_flow}, the procedure $\WitRoute$ is used to decompose extension arcs into abundant arc paths that had already existed at the time of an extension arc's creation.
This approach is taken by~\cite{itco_flow} by extending the ITCO data structure to maintain a representation of paths corresponding to transitive arcs (the so-called witness list).
But we could in fact decompose each extension arc into the earliest abundant arc path connecting its endpoints. Consequently, instead of maintaining paths in the incremental transitive closure data structure, we can instead compute the decomposed of extension arcs into valid paths only at the very end.
In the following, we provide more details of this idea and argue it parallelizes rather effortlessly.

Let us create a graph $G = (V, E')$ which consists of vertices $V$ from the original instance and the set of abundant arcs $E'$ present at the end of the algorithm, excluding the shortcut arcs $\Fs\cup\rev{\Fs}$. We label each arc by the iteration in which it became abundant and run Dijkstra's algorithm from each vertex using the maximum operation in place of addition. This change will cause the algorithm to compute the paths with the smallest maximum label, i.e., one of the earliest paths, instead of the shortest paths. For an extension arc added in iteration $k$, its certifying abundant path has maximum label less than $k$; hence, the selected tree path is simple, has maximum label less than $k$, and contains only original arcs and orientations of extension arcs added before iteration~$k$.
Dijkstra's algorithm can be implemented in $\Ot(m)$ work and $\Ot(n)$ depth using batch-parallel priority queues; see, e.g.,~\cite{BrodalTZ98}.

After constructing the optimal path trees, we can decompose each of the $O(n)$ extension arcs in parallel.
For each extension arc $(a, b)$, using the optimal path tree rooted at vertex $a$, we can decompose the arc into a list of abundant arcs on the path $a \to b$. This can be done by iteratively traversing the tree upward from vertex $b$ to the root $a$.

We then iterate over the extension arcs in reverse chronological order of their insertion. For each extension arc $e$, we move the signed amount $f_e-f_{\rev{e}}$ to its decomposed path and set $f_e=f_{\rev{e}}=0$. The reverse order is necessary, as the abundant arcs used in the optimal path trees can be orientations of earlier extension arcs. Each iteration takes $\Ot(n)$ work and $\Ot(1)$ depth, yielding $\Ot(n^2)$ total work and $\Ot(n)$ total depth for moving the flow. Replacing an extension-arc pair by its signed net flow on a path with the same endpoints preserves every vertex excess and the flow value; a backward induction using the reverse order therefore shows that this eliminates all extension arcs without reintroducing a processed one. Moreover, the selected paths are simple and consist of arcs that became abundant before the replaced extension arc was inserted, so the flow-change bound and the resulting capacity-feasibility argument in the proof of~\cite[Lemma~6.9]{itco_flow} apply unchanged.

\cite[\textsf{RerouteShortcut}]{itco_flow} already takes $\Ot(n)$ work.

\subsection{Work and depth analysis}
\begin{theorem}
   The parallel version of~\Cref{alg:max-flow} has $\Ot(nm)$ work and $\Ot(m_c)$ depth.
\end{theorem}
\begin{proof}
   In \Cref{sec:parallelizing_basics} we have shown that most of the parts of the algorithm satisfy the required work and depth bounds because that section relies mostly on the assumption from~\Cref{lem:essential_bound} that $\sum_i m_i = O(m_c)$, where $m_i$ is the number of essential arcs in iteration $i$.

Recall from~\Cref{cor:iterations} that the algorithm performs $O(m_c/n^{1/3})$ iterations, and thus $\dsTcAdd$ is called $O(m_c/n^{1/3})$ times. $\dsTcAdd$ is implemented via issuing a batch of insertions to the incremental transitive closure data structure. Hence, by~\Cref{thm:inc-tc}, maintaining the ITCO data structure requires $\Ot(nm)$ work and $\Ot(m_c)$ total depth.

   Recall that the incremental transitive closure-based ITCO implementation~\cite{itco_flow} used in the algorithm returns, for a query $S\subseteq V$, either a graph with vertices $S$ and $O(|S|^2)$ arcs if $|S|<\sqrt{m}$, or the entire graph $G$ otherwise. Let $m_i$ be the number of essential arcs in iteration $i$ and let $\hat{n}_i, \hat{m}_i$ be the numbers of vertices and arcs in the auxiliary graph in iteration $i$.
   Thus, either $\hat{n}_i = O(m_i)$ and $\hat{m}_i = O(m_i^2)$ if $m_i \le \sqrt{m}$ or $\hat{n}_i = O(n)$ and $\hat{m}_i = O(m)$ if $m_i > \sqrt{m}$. In both cases the auxiliary graph $\bar{G}$ has $m_i$ finite-capacity arcs.

   In $O(m_c / \sqrt{m})$ iterations where $m_i > \sqrt{m}$,
   the total work used by~$\approxFlow$ is
   \[\sum_i \Ot(\hat{m}_i + \hat{n}_i^{1.5} + m_in) \le \Ot(m + n^{1.5}) \cdot \Ot(m_c / \sqrt{m}) + \Ot(m_cn) \le \Ot(nm).\]
   The total depth of these iterations is \[\sum_i \Ot(\sqrt{\hat{n}_i} + m_i) \le \Ot(\sqrt{n}) \cdot \Ot(m_c / \sqrt{m}) + \Ot(m_c) = \Ot(m_c).\]
   For iterations where $m_i \le \sqrt{m}$,
   the total work is
   \[ \sum_i\Ot(\hat{m}_i + \hat{n}_i^{1.5} + m_in) = \sum_i\Ot(m_i^2 + m_i^{1.5}) + \Ot(m_cn) \le \Ot(m_c\sqrt{m}+m_cn)\leq \Ot(nm),\] whereas the total depth is \[\sum_i\Ot(\sqrt{\hat{n}_i} + m_i) \le \sum_i\Ot(m_i + m_i) = \Ot(m_c).\]

   $\ExactFlow$ is used only for iterations where $m_i=O(n^{1/3})$. Because there are $m_i = \Omega(\hat{n}_i)$ finite-capacity arcs, the total work is
   \[\sum_i \Ot(\hat{n}_i^3m_i) \le \sum_i \Ot(nm_i) \le \Ot(nm_c).\]
   The total depth of these iterations is $\sum_i \Ot(m_i) = \Ot(m_c)$.
   Let $N=n+m$ for the original input. By choosing the constants in the high-probability guarantees appropriately, we invoke every $\ExactFlow$ and $\approxFlow$ call with failure probability at most $N^{-(\gamma+1)}$. There are $O(m_c/n^{1/3})=O(N)$ such calls in total, so a union bound shows that they all succeed with probability at least $1-O(N^{-\gamma})$.
\end{proof}

\section{Parallel Incremental Transitive Closure}\label{sec:hitting_set_itco}
In this section we prove the following theorem, which is a more detailed version of~\Cref{thm:inc-tc}.

\begin{theorem}\label{theorem:full_parallel_closure}
Let $G=(V,E)$ be a directed graph subject to batches of arc insertions. There exists a Monte Carlo randomized data structure maintaining the transitive closure in $\Ot(nm)$ total work, where $m\geq n$ is the final number of arcs inserted.

Specifically, the transitive closure matrix is maintained explicitly by the data structure.

Every batch of insertions is processed within $\Ot(n^{1/3})$ depth.
The answers produced are correct with high probability.
\end{theorem}
This data structure can be used to implement the ITCO data structure from \cite{itco_flow} (see~\Cref{def:itco}), as described in~\Cref{sec:alg_overview}.

\subsection{Data structure for dense graphs}

Let us start with a simple data structure (resembling the incremental approximate APSP data structure of~\cite{DBLP:conf/soda/KarczmarzL20}) that is nevertheless very efficient for dense graphs. It will be useful later.
\begin{lemma}[Parallel incremental closure]\label{lemma:parallelclosure}
   Let $G=(V,E)$ be a directed graph subject to batches of arc insertions such that the total number of arcs inserted is $O(n^2)$. There exists a data structure that explicitly maintains a reachability matrix of $G$ with total work $\Ot(n^3)$ such that each batch of insertions is processed in $\Ot(1)$ depth.
\end{lemma}
\begin{proof}
  We maintain the reachability in $G$ as a transitively closed matrix $M$, that is, if $M_{a,b} = 1$ and $M_{b, c} = 1$, then $M_{a, c} = 1$. For each inserted arc $(i, j)$, we set $M_{i, j} := 1$, and we update the matrix to fix the transitivity. This matrix will satisfy $M_{i,j} = 1$ iff there is a path from $i$ to $j$; as such, we will refer to $M_{i, j} = 1$ as a path/arc from $i$ to $j$. We will use $(i, j)$ to denote a path (possibly a single arc) from $i$ to $j$.

To achieve transitivity, when adding a set of arcs, we check whether $M_{i, j} = 1$ for each arc $(i, j)$. If this is the case, we ignore the arc; otherwise, we update the matrix by marking $M_{i, j} := 1$, and we add the pair $(i, j)$ to the set $D_0$. $D_i$ denotes the set of newly discovered paths in iteration $i$. After iteration $i$, each path $(j, k) \in D_i$ will be marked by setting $M_{j, k} := 1$. In each iteration we will explore the missing paths based on paths discovered in the previous iteration until we reach an iteration which does not discover any new paths.

In iteration $i$, in parallel, for each path $(a, b) \in D_{i-1}$ and each $c \in V$, we know that $M_{a, b} = 1$, and we check whether $M_{b, c} = 1 \land M_{a, c} = 0$ holds. If so, we have discovered a new path $(a,c)$. We proceed similarly if $M_{c, a} = 1 \land M_{c, b} = 0$, discovering a new path $(c,b)$.
After discovering a new path, we add it to $D_i$.
At the end of the iteration, we mark $M_{a, b} := 1$ for all $(a, b) \in D_i$.

We will show that if the matrix $M$ is initially transitively closed, then after inserting some arcs and after $k$ iterations the matrix satisfies $M_{a,b} = 1$ for all pairs $(a,b)$ such that there is a path $a\to b$ of length at most $2^k$ in the graph.

We will prove this fact by induction. $k = 0$ is trivial, as before the first iteration the matrix contains the previous arcs and the newly inserted arcs. For $k + 1$, let us assume we have a path $(a, b)$ with length at most $2^{k+1}$ such that before iteration $k+1$ it is not included in the matrix. Let us divide this path into two paths $(a, m)$ and $(m, b)$, both of length at most $2^k$. By induction, before iteration $k + 1$, we must have had $M_{a, m} = 1$ and $M_{m, b} = 1$. If $(a, m) \notin D_i$ and $(m, b) \notin D_i$ for any $i = 0, 1, \cdots, k$, then before adding any new arcs it must have been true that $M_{a, m} = 1, M_{m, b} = 1, M_{a, b} = 0$, which contradicts the assumption that the matrix was transitively closed. This means that there was a moment when one of the paths $(a, m)$ or $(m, b)$ was in some $D_i$ for $i \le k$ and the other path was in $M$. This, however, contradicts the algorithm as the algorithm would have discovered the path $(a, b)$.

Each iteration can be implemented in $\Ot(1)$ depth because we can explore each pair $(e,k) \in D_i\times V$ in parallel and then deduplicate the results in $\Ot(1)$ depth. There are $\Ot(1)$ iterations, meaning each batch insertion takes $\Ot(1)$ depth. Each path from $i$ to $j$ can be marked once in $M$, meaning it can be in at most one of the sets $D_i$ throughout all the arc insertion operations. For each of the $O(n^2)$ paths we perform $\Ot(n)$ work, bounding the total work by $\Ot(n^3)$.
\end{proof}

\subsection{Data structure for sparse graphs}

We define multiple data structures one by one as building blocks for the final data structure.
\begin{lemma}\label{lemma:short_path_single_closure}
Let $G=(V,E)$ be an initially empty directed graph subject to batches of arc insertions and let $v \in V$. There exists a data structure that maintains a subset $A\subseteq V$ such that (1) $A$ contains only vertices reachable from $v$ in $G$, and (2) all vertices $v'$ such that the distance from $v$ to $v'$ is at most $n^{1/3}$ are contained in $A$.
The total work of the data structure is $\Ot(m)$. Each batch insertion is processed in $\Ot(n^{1/3})$ depth.
\end{lemma}
\begin{proof}
To implement this data structure, we can use a parallel version of truncated BFS.

We will use the set of vertices visited by BFS as the maintained set $A$ from the lemma.
Throughout the algorithm we will maintain a set $S$ of vertices that are reachable from $v$ but are not marked in the visited vector.
Specifically, $S$ contains the ``frontier'' of the truncated BFS run, i.e., vertices that are not yet in $A$ but are reachable from $A$ via a single arc.

Initially, we set $S := \{v\}$. After each batch insertion, we run a number of iterations that work as follows.
In each iteration we mark vertices in $S$ as visited (i.e., move them to $A$) and compute the set $N$ consisting of unvisited neighbours of $S$. At the end of an iteration we set $S := N$.
We stop when $S = \emptyset$ or after $n^{1/3} + 1$ iterations.
Note that an iteration takes $\Ot(1)$ depth and work proportional to the current degree of the vertices in $S$.
After stopping, we are guaranteed to have visited every vertex reachable within $n^{1/3}$ arcs from $S$ as it stood before the batch was issued. After the algorithm finishes, we keep the value of $S$.

To process a batch of arc insertions, first, for each inserted arc $(a, b)$ in parallel, if $a$ is marked as visited and $b$ is not, we add $b$ to $S$ (if it is not already there).
This step takes $\Ot(1)$ depth and work linear in the size of the batch.
Then we perform at most $n^{1/3} + 1$ iterations, as described above.
This guarantees that every vertex at distance at most $n^{1/3}$ from any of the visited vertices, including $v$, becomes visited.

The above analysis implies that the total work is $\Ot(m)$, and the depth per insertion is $\Ot(n^{1/3})$.
\end{proof}
Using $n$ copies of this data structure for all vertices $v$ of the graph, we immediately get the following data structure:
\begin{lemma}[Short Path Data Structure]\label{lemma:short_path_closure}
   Let $G=(V,E)$ be a directed graph subject to batches of arc insertions. There exists a data structure that maintains a reachability matrix $M$ with total work $\Ot(nm)$ and depth $\Ot(n^{1/3})$ per batch insertion such that (1) $M_{i, j} = 1$ only if there exists a path in $G$ from $i$ to $j$, and (2) if there exists a path in $G$ from $i$ to $j$ with length at most $n^{1/3}$, then $M_{i, j} = 1$.
\end{lemma}

For the next data structures we will need the following folklore lemma (see, e.g.,~\cite{DBLP:journals/siamcomp/UllmanY91}).
\begin{lemma}[Hitting Set]\label{lemma:hitting_set}
   Let $S$ be a set of size $n$ and $\mathcal{S}$ be a family of subsets of $S$ such that $|\mathcal{S}| = \poly(n)$ and $|X| \ge k$ for $X \in \mathcal{S}$. Let us take $H \subset S$ of size $\Theta((n / k)\log{n})$ sampled uniformly at random. Then with high probability, for all $X \in \mathcal{S}$, $X \cap H \neq \emptyset$. We call $H$ a \emph{hitting set}.
\end{lemma}
\begin{proof}
   Let $|\mathcal{S}| \le n^{\alpha}$ and $|H| = (c + \alpha)n/k\ln n$. We will assume that $|H| < n$. The probability that $X \cap H = \emptyset$ for any fixed $X \in \mathcal{S}$ is at most
   \[\prod_{i=0}^{|H|-1} \left(\frac{n - k - i}{n-i}\right) \le \left(\frac{n - k}{n}\right)^{|H|} = \left(1 - \frac{k}{n}\right)^{(c + \alpha)n/k\ln n} \le e^{-(c + \alpha)\ln n} = n^{-(c + \alpha)}.\]
   Using the union bound, we get that the probability that $X \cap H = \emptyset$ for some $X \in \mathcal{S}$ is at most
   \[|\mathcal{S}| \cdot n^{-(c + \alpha)} \le n^{-c}.\]
\end{proof}

Define a collection of paths $\mathcal{P}$ as follows.
Suppose $v$ becomes reachable from $u$ for the first time after the $j$-th batch insertion, and let $G_j$ be the graph $G$ at that time.
Fix some arbitrary simple $u\to v$ path $P_{u,v}$ in $G_j$.
Then let $\mathcal{P}$ include all the subpaths of $P_{u,v}$ on $n^{1/3}$ vertices, if there are any. This way, $\mathcal{P}$ contains at most $n^3=\poly(n)$ subpaths.
Hence, a hitting set $H$ of size $\Theta(n^{2/3}\log{n})$ obtained as described in~\Cref{lemma:hitting_set} for $S=V$, $k=n^{1/3}$, and $\mathcal{S}=\mathcal{P}$, contains a vertex of each subpath from $\mathcal{P}$ with high probability.

\begin{lemma}[Hitting Set Path Data Structure]\label{lemma:h_h_closure}
Let a digraph $G=(V, E)$ be subject to batches of arc insertions. There exists a data structure maintaining (with high probability) the reachability matrix in $G$ for pairs of vertices in $H$ with $\Ot(nm)$ total work. The depth per batch insertion is $\Ot(n^{1/3})$.
\end{lemma}
\begin{proof}
   Let us set up the data structure from \Cref{lemma:short_path_closure} maintaining a matrix $M$, as defined in that lemma.

   For any two vertices $h_1, h_2 \in H$,
   consider an earliest path $P_{h_1,h_2}$, defined previously.
   With high probability, this path can be divided into sub-paths of length at most $n^{1/3}$ between vertices of $H$. Otherwise, there would be a sub-path with $n^{1/3}$ vertices without an element from the hitting set, but this would be a contradiction as this sub-path is in $\mathcal{P}$.
   As a result, if we consider a graph $G'=(H,E')$ where there is an arc $uv$ if and only if $M_{u,v}=1$, we conclude that a path $h_1\to h_2$ exists in that graph with high probability.
   Clearly, a path $h_1\to h_2$ in $G'$ also implies that one exists in $G$.

   This means that by using the data structure from \Cref{lemma:parallelclosure} on the (incremental) graph $G'=(H,E')$, we can maintain reachability in $G$ among the vertices of $H$.

   The data structure from \Cref{lemma:short_path_closure} takes $\Ot(nm)$ total work; the data structure from \Cref{lemma:parallelclosure} takes $\Ot((n^{2/3})^3) = \Ot(n^2)$ total work, and both data structures use $\Ot(n^{1/3})$ depth per batch insertion.
\end{proof}
\begin{lemma}[Shortcut Data Structure]\label{lemma:h_v_closure}
   Let a digraph $G=(V, E)$ be subject to batches of arc insertions. There exists a data structure maintaining, with high probability, the reachability relation for all pairs $H\times V$ in $G$ with $\Ot(nm)$ total work and $\Ot(n^{1/3})$ depth per batch insertion.
\end{lemma}
\begin{proof}
   We use the data structure from \Cref{lemma:h_h_closure} to maintain the reachability matrix for vertices from the hitting set $H$. We then use $|H|$ copies of the data structure from \Cref{lemma:short_path_single_closure}, each corresponding to a vertex from the hitting set. For the copy corresponding to vertex $h$, we use an extended graph $G_h=(V,E_h)$ where $E_h$ is a set of arcs from $E$ plus ``shortcut'' arcs $(h, h')$ representing all paths from $h$ to $h' \in H$. Inside the data structure from \Cref{lemma:h_h_closure} we track new entries in the reachability matrix after each arc insertion to update sets of arcs $E_h$. Let $h\in H$, $v\in V$, and consider the earliest path $h\to v=P_{h,v}$, defined previously. Let us assume that this path has length greater than $n^{1/3}$. Consider a suffix of this path on $n^{1/3}$ vertices. This suffix is contained in $\mathcal{P}$, meaning that there exists a vertex $h' \in H$ in the suffix. This implies that in $G_h$ there exists a path $h\to h'\to v$ with length at most $n^{1/3}$.

   $|H|$ copies of the data structure from \Cref{lemma:short_path_single_closure} take total work $\Ot(|H|m) \le \Ot(nm)$, and all copies can be updated in parallel in $\Ot(n^{1/3})$ depth.
\end{proof}
We can now prove \Cref{theorem:full_parallel_closure}.
\begin{proof}
   The construction is roughly the same as in \Cref{lemma:h_v_closure} except that we use \Cref{lemma:h_v_closure} on a reversed graph to maintain the shortcuts from $V$ to $H$.

   Instead of using a copy of the data structure from \Cref{lemma:short_path_single_closure} for each vertex of the hitting set, we use $n$ copies, one for each vertex $v \in V$, and add $|H|$ shortcuts from $v$ to the corresponding copy.

   The total work from $n$ copies of the data structure from \Cref{lemma:short_path_single_closure} is $\Ot(nm)$. The rest of the argumentation is the same.
\end{proof}

\section{Removing \texorpdfstring{$(\cdot, s)$}{(., s)} and \texorpdfstring{$(t, \cdot)$}{(t, .)} arcs in \texorpdfstring{$\Ot(m)$}{\~O(m)} work and \texorpdfstring{$\Ot(n)$}{\~O(n)} depth}\label{sec:removing_v_s_and_t_v}
Recall \Cref{sec:preliminaries}. The algorithm of Dadush et al.~\cite{DBLP:conf/soda/DadushOSV26} assumes the existence of distinguished additional infinite-capacity arcs $(v, s)$ and $(t, v)$ for all $v \in V\setminus\{s,t\}$. If we were to decompose a maximum flow into path flows and cycle flows, the flow on these auxiliary arcs would be contained in cycle flows, meaning that these arcs do not change the maximum flow value of the instance and can be removed in $\Ot(m)$ time~\cite{DBLP:journals/jcss/SleatorT83}, which is insufficient to guarantee $\Ot(m_c)$ depth.

In this section we will prove the following theorem.
\begin{theorem}\label{thm:additional_arc_removal}
   There exists an algorithm that, given a feasible flow, modifies it so that the flow on the distinguished auxiliary arcs $(v, s)$ and $(t, v)$ for all $v \in V\setminus\{s,t\}$ is equal to 0 and the flow's value is at least the original value. The algorithm takes $\Ot(m)$ work and $\Ot(n)$ depth.
\end{theorem}

\subsection{Description of the algorithm and the proof of correctness}
We will construct a simpler algorithm, similar to the algorithm described in \Cref{thm:additional_arc_removal}, except that it only guarantees that the flow on the distinguished arcs $(v, s)$ is zero. By reversing the flow direction, using this algorithm with the roles of $s$ and $t$ temporarily swapped, and reversing the direction of the resulting flow back, we can remove the flow from the distinguished arcs $(t, v)$. To then remove the resulting flow from the distinguished arcs $(v, s)$, we can run the algorithm again, this time without any changes to the flow, giving us \Cref{thm:additional_arc_removal}. As shown below, each invocation of the simpler algorithm only decreases the flow on arcs of its input instance, so the second invocation cannot reintroduce flow on the distinguished arcs $(t,v)$ whose flow was set to zero by the first invocation.

At the beginning of the simpler algorithm, we set the flow on every self-loop to 0 and remove all self-loops from the working instance. This does not change any vertex excess or the flow value. At the end, the input self-loops are restored with flow 0.

Let us define \emph{vertex contraction} for a vertex $v \in V \backslash \{s, t\}$ as follows. Assume we are given two lists: the list of incoming arcs with their respective flow values and the list of outgoing arcs with their flow values. Because $\ex(v) = 0$, the sums of the flows in the first list and in the second list are the same. Let us assume that the first elements of those lists are arc objects $a = (u, v)$ with flow $f_a$ for the first list and $b = (v, w)$ with flow $f_b$ for the second list, and let $x = \min\left(f_a, f_b\right)$. If $u\neq w$, we remove $x$ units of flow from arcs $a$ and $b$ and add them to a fresh arc object $c = (u, w)$. We call such an operation a \emph{small contraction}. If $u=w$, we instead remove $x$ units of flow from $a$ and $b$ without creating a self-loop; we call this operation a \emph{two-cycle cancellation}. After updating the flow values in the lists, at least one of the arcs $a$ and $b$ will have zero flow and thus can be removed from the list.

For each small contraction we also remember a \emph{small contraction entry}, which stores the arc $c$ itself, the amount of flow we have added to that arc, and the two arcs $a$ and $b$ from which $c$ was composed. With such information we are able to revert the small contraction using the corresponding entry. We do not create an entry for a two-cycle cancellation. We repeat the process until both lists are empty. When comparing flows before and after a small contraction, we regard the fresh arc $c$ as having flow 0 before its creation.

After this process ends (that is, both of the lists are empty), vertex $v$ has no incident arcs and can be removed from the instance, finalizing the vertex contraction.

\begin{lemma}\label{lem:contraction_excess}
   After each small contraction, two-cycle cancellation, or vertex contraction, the resulting flow $f'$ satisfies $\ex_{f'}(v) = 0$ for $v \in V \backslash \{s, t\}$ and has the same value as the original flow.
\end{lemma}
\begin{proof}
   Every time we remove $x$ units of flow from arcs $(u, v), (v, w)$ and add $x$ units of flow to $(u, w)$ in a small contraction, the excesses of $u$, $v$ and $w$ do not change. A two-cycle cancellation removes a circulation, so it also preserves all excesses. Every vertex contraction is a list of small contractions and two-cycle cancellations.
\end{proof}

\begin{lemma}
   After each small contraction, two-cycle cancellation, or vertex contraction, the number of arcs with positive flow does not increase.
\end{lemma}
\begin{proof}
   For each small contraction we can possibly create a new arc $(u, w)$, but we remove all flow from either arc $(u, v)$ or $(v, w)$ (or both), removing at least one arc. A two-cycle cancellation creates no arc and also empties at least one arc. Thus, the number of positive-flow arcs cannot increase. The same is true for a vertex contraction.
\end{proof}

\begin{definition}[Contraction reversal]
   Let us assume that we are given some flow $f$ and a small contraction entry
   which says that $x$ units of flow were moved from $a$ and $b$ to $c$. Reverting the small contraction is defined as removing $x' = \min\left(x, f_c\right)$ units of flow from $f_c$ and adding $x'$ units of flow to $f_a$ and $f_b$.

   Reverting a vertex contraction is defined as reverting all of its small contractions in reverse chronological order. Two-cycle cancellations have no entries and are not reverted.
\end{definition}

Note that in the definition above we transfer $x'$ units of flow instead of $x$ because, in the algorithm, we will revert the contraction on some flow different to the result of the contraction.

\begin{lemma}\label{lem:reverting_contraction}
   Reverting a small or vertex contraction results in a flow that satisfies $\ex(v) = 0$ for $v \in V \backslash \{s, t\}$ and has the same value as the original flow.
\end{lemma}
\begin{proof}
   The proof is the same as for \Cref{lem:contraction_excess}. The net change in each excess is $0$.
\end{proof}

\begin{lemma}\label{lem:contraction_invariant}
   Let us assume that a small contraction or a two-cycle cancellation applied to a flow $f$ results in a flow $g$.
   Let us take any flow $g' \le g$. Reverting the small contraction using $g'$ as the flow instead of $g$, or doing nothing for a two-cycle cancellation, produces a flow $f'$ that satisfies $f' \le f$.
\end{lemma}
\begin{proof}
   For a two-cycle cancellation, $g\le f$, so doing nothing gives $f'=g'\le g\le f$. Now, let us take a small contraction entry that moves $x$ units of flow from arcs $a$ and $b$ to the fresh arc $c$. If $x \le g'_c$, then $x' = x$ when we revert the contraction, meaning the change to both flows, $g'$ and $g$, is the same, so $f' \le f$.

   In the case where $x > g'_c$, let us look at the changes made by reverting the small contraction for both $g$ and $g'$. For $g$, we increase the flow on arcs $a$ and $b$ by $y = \min\left(x, g_c\right)$ and decrease the flow on arc $c$ by $y$. For $g'$, reverting the contraction adds
   \[y' = \min\left(x, g'_c\right) = g'_c \le y\]
   units of flow to arcs $a$ and $b$, and sets the flow on arc $c$ to 0. Because in $f'$ we have added back at most as much flow as in $f$, and $f'_c = 0 = f_c$, this implies that after those changes we have $f' \le f$.
\end{proof}

By cleverly using the vertex contractions, we can remove the flow from the distinguished auxiliary arcs $a_v = (v, s)$. We construct a graph in which, for each vertex $v \in V \backslash \{s, t\}$, we create an additional vertex $\hat{v} \in \hat{V}$. We replace each arc $a_v$ with two fresh arcs $p_v = (v, \hat{v})$ and $q_v = (\hat{v}, t)$, both with the same flow as $a_v$. This temporarily increases the flow's value. We can also assume w.l.o.g. that the input graph has no flow on any arc from $t$ to $s$, since otherwise we could set the flow on every such arc to zero to improve the flow.

We then iteratively take a vertex $v \in V \backslash \{s, t\}$ with the minimum current number of incident positive-flow arcs, contract it and repeat until the only vertices that are left are $s, t$ and $\hat{V}$.

We will now prove that contractions and their reversals maintain the following invariant.

\begin{lemma}\label{lem:invariant_additional_vertices_outputs}
   Small contractions and their small contraction reversals preserve the invariant that, for every $\hat{v} \in \hat{V}$, the flow on arc $q_v$ is never modified and $f_e = 0$ for every other arc object $e$ with $\operatorname{tail}(e) = \hat{v}$.
\end{lemma}
\begin{proof}
   Let us assume otherwise. If a positive-flow arc $(\hat{v}, x)$ was the result of some small contraction, then this means that before that contraction there were positive-flow arcs $(\hat{v}, w), (w, x)$ for some $w \in V \backslash \{s, t\}$, but the existence of the arc $(\hat{v}, w)$ contradicts the invariant.

   Let us assume that the arc $(\hat{v}, x)$ was the result of some small contraction reversal. The first case is that the small contraction reversal turned some arc $(\hat{v}, z)$ for some $z$ into arcs $(\hat{v}, x), (x, z)$. This would imply that the original small contraction belongs to the vertex contraction of $x \in V \backslash \{s, t\}$. Because of the invariant and because $(\hat{v}, x)$ had positive flow at the time of the contraction, we get $x = t$, which contradicts $x \in V \backslash \{s, t\}$.
   The second case is that the small contraction reversal turned an arc $(z, x)$ for some $z$ into arcs $(z, \hat{v}), (\hat{v}, x)$, but this implies that the small contraction belongs to the vertex contraction of $\hat{v}$, which again is a contradiction.
\end{proof}

Because the invariant is true before contracting vertices, \Cref{lem:invariant_additional_vertices_outputs} implies that it is also true afterward. Because the graph at the end has vertex set $\{s,t\}\cup\hat{V}$, the invariant implies that the only arcs that can be in the graph are $(s, t)$, $(t, s)$, $(s, \hat{v})$, $q_v=(\hat{v}, t)$, and $(t, \hat{v})$ for all $\hat{v} \in \hat{V}$.

Let us set the flow on all arcs incident to vertices in $\hat{V}$ to 0. This change preserves $\ex(v) = 0$ for $v \notin \{s, t\}$. Let $q$ be the value of the original flow after potentially removing the flow on arcs from $t$ to $s$, and let $C = \sum_{v \in V\setminus\{s,t\}} f_{q_v}$, where these flow values are taken immediately before this change. By \Cref{lem:invariant_additional_vertices_outputs}, $C$ is unchanged by the contractions and is exactly the amount by which replacing the arcs $a_v$ increased the flow value. Thus, the flow immediately before removing the arcs incident to $\hat{V}$ has value $q+C$.

Let $B = \sum_{e:\,\operatorname{tail}(e)=t,\,\operatorname{head}(e)\in\hat{V}} f_e$ immediately before this removal. Removing the arcs incident to $\hat{V}$ changes the excess of $t$ by $-C+B$, so the resulting flow has value $q+B\ge q$, which is at least the value of the input flow. Let $f'^{(n - 2)}$ be the flow after this change, where $n - 2$ is the number of vertex contractions performed by the algorithm, and let $f^{(i)}$ be the original flow after $i$ contractions performed by the algorithm.

We can now revert each vertex contraction one by one. We start with flow $f'^{(n - 2)}$, revert the last ($(n - 2)$-th) contraction to get the flow $f'^{(n - 3)}$, then revert the second-to-last ($(n - 3)$-th) contraction to get the flow $f'^{(n - 4)}$ and so on, until we get the flow $f'^{(0)}$ after reverting all the vertex contractions. Because $f'^{(n - 2)} \le f^{(n - 2)}$, applying \Cref{lem:contraction_invariant} $n - 2$ times gives $f'^{(0)} \le f^{(0)}$. Because $f^{(0)} \le u$, we get that $f'^{(0)}$ is a feasible flow.

\begin{lemma}
   The flow $f'^{(0)}$ does not use the additional arcs $q_v$ or any arc object with head in $\hat{V}$.
\end{lemma}
\begin{proof}
   By \Cref{lem:invariant_additional_vertices_outputs} and the facts that the invariant is satisfied in $f'^{(n - 2)}$ and that $f'^{(n - 2)}_{q_v} = 0$ for $\hat{v} \in \hat{V}$, we get that $f'^{(0)}_{q_v} = 0$. Because $\ex(\hat{v}) = 0$ and every other arc leaving $\hat{v}$ also has zero flow, we get that $f'^{(0)}_e = 0$ for every arc object $e$ with $\operatorname{head}(e) = \hat{v}$.
\end{proof}

The resulting flow does not use the distinguished arcs $a_v$, does not use newly added arcs and has at least the same value as the original one. Moreover, the inequality $f'^{(0)}\le f^{(0)}$ proved above implies that, on every arc of the original instance, the resulting flow is at most the input flow.

\subsection{Parallel implementation}
In this section we will show how to implement the algorithm in $\Ot(m)$ work and $\Ot(n)$ depth, as in \Cref{thm:additional_arc_removal}.

\begin{lemma}
   Each vertex contraction can be performed in $\Ot(1)$ depth.
\end{lemma}
\begin{proof}
   Let us visualize how we compute the vertex contraction. In \Cref{fig:contraction} we have drawn two lists, one of incoming arcs and one of outgoing arcs, in the following way. Let $x_i$ be the flow values in the first list and $y_i$ be the flow values in the second list. We draw a main interval that consists of a sub-interval of length $x_1$, then a sub-interval of length $x_2$ and so on. We do the same for the other list using $y_i$. Both main intervals have the same total length, because $\ex(v) = 0$.

   The process of a single vertex contraction can be interpreted as follows. We go from left to right on the main intervals until one of the sub-intervals ends. We draw there a vertical dotted line and repeat the process until we get to the end of the main intervals. We also add dotted lines at the beginning and end of the main intervals.

   The small contractions that are the result of the vertex contraction are represented by the segments between the dotted lines. Arcs from the sub-intervals that intersect the segment correspond to arcs that the small contraction combines, and the length of a segment corresponds to the amount of flow moved by the small contraction onto the combined arc.
   \begin{figure}[!ht]
      \centering
      \begin{tikzpicture}[scale=0.7]
      \draw (0,0) -- (12,0);
      \draw (0,0.5) -- (0,-0.5);
      \draw (1, 0.5) node{$y_1$};
      \draw (2,0.5) -- (2,-0.5);
      \draw (4, 0.5) node{$y_2$};
      \draw (6,0.5) -- (6,-0.5);
      \draw (7.5, 0.5) node{$y_3$};
      \draw (10,0.5) -- (10,-0.5);
      \draw (11, 0.5) node{$y_4$};
      \draw (12,0.5) -- (12,-0.5);

      \draw (0,2) -- (12,2);

      \draw (0,1.5) -- (0,2.5);
      \draw (1.5, 2.5) node{$x_1$};
      \draw (3,1.5) -- (3,2.5);
      \draw (5, 2.5) node{$x_2$};
      \draw (7,1.5) -- (7,2.5);
      \draw (8.5, 2.5) node{$x_3$};
      \draw (10,1.5) -- (10,2.5);
      \draw (11, 2.5) node{$x_4$};
      \draw (12,1.5) -- (12,2.5);

      \draw (0, 3) edge[dotted] (0, -1);
      \draw (2, 3) edge[dotted] (2, -1);
      \draw (3, 3) edge[dotted] (3, -1);
      \draw (6, 3) edge[dotted] (6, -1);
      \draw (7, 3) edge[dotted] (7, -1);
      \draw (10, 3) edge[dotted] (10, -1);
      \draw (12, 3) edge[dotted] (12, -1);
      \end{tikzpicture}
      \caption{Visual representation of a vertex contraction}
      \label{fig:contraction}
   \end{figure}
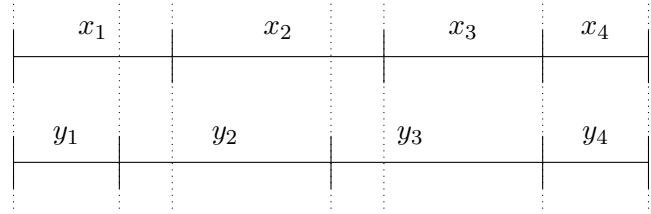
   We can compute the set of positions of the dotted lines by computing the prefix sums for both of the lists~\cite{parallel_scan}, combining them, adding positions at the beginning and end of the main interval, sorting the positions~\cite{parallel_sort} and removing duplicates. Using binary search on the computed prefix sums, for each segment between consecutive positions, we can find the incoming and outgoing arcs to which the segment corresponds. Each small contraction adds flow to a different arc, whereas segments whose incoming and outgoing arcs form a two-cycle are discarded without creating an arc. In total, the operations remove all of the flow from all the arcs incident to the contracted vertex. Those changes can be trivially implemented in $\Ot(1)$ depth.
\end{proof}
\begin{lemma}
   Each vertex contraction reversal can be done in $\Ot(1)$ depth.
\end{lemma}
\begin{proof}
   Each small contraction reversal only adds some flow to some arcs and removes the same amount of flow from some arc. Two-cycle cancellations have no entries to reverse. Because each remaining small contraction reversal for a given vertex contraction decomposes a different arc, we can compute each $x'$ independently. Performing a set of $d$ increments/decrements in parallel can be done in $\Ot(d)$ work and $\Ot(1)$ depth.
\end{proof}

\begin{corollary}
   The algorithm runs in $\Ot(n)$ depth.
\end{corollary}
\begin{proof}
   There are $O(n)$ contractions and contraction reversals; each one takes $\Ot(1)$ depth.
\end{proof}

\begin{lemma}
   All vertex contractions and vertex contraction reversals can be done in $\Ot(m)$ total work.
\end{lemma}
\begin{proof}
   For a vertex $v$ that has $d_v$ incident positive-flow arcs, the contraction takes $\Ot(d_v+1)$ work. Reverting such a contraction also takes $\Ot(d_v+1)$ work. We maintain the vertices that can be removed in a parallel BST~\cite{BlellochFS16,GilMV91,PaulVW83}, keyed by their current pairs $(d_v,v)$. A contraction changes $O(d_v)$ arc incidences; after aggregating these changes by endpoint, we can update the affected keys in one batch. Thus, maintaining the BST and choosing its minimum takes $\Ot(d_v+1)$ work and $\Ot(1)$ depth per contraction. Because $\sum_{v \in V \cup \hat{V}}d_v = 2\hat{m} = O(m)$, where $\hat{m}$ is the current number of positive-flow arcs in the constructed instance, we know that if at a given moment there are $\hat{n}$ vertices that we can remove, then the smallest $d_v$ among these is at most $2\hat{m}/\hat{n}$. If this were not the case, then we could lower bound the sum of $d_v$ by more than $2\hat{m}$.

   For the first contraction, the algorithm can choose to remove one of $|V \backslash \{s, t\}| = n - 2$ vertices; for the next one, it can choose one of $n - 3$ vertices; and so on, until only one vertex is available for the last contraction. We can bound the sum of $d_v$ over the chosen vertices by
   \[\frac{2\hat{m}}{n - 2} + \frac{2\hat{m}}{n - 3} + \dots + \frac{2\hat{m}}{1} = O(m)O(\log n) = \Ot(m).\]
   The total work of vertex contractions and vertex contraction reversals can be bounded by
   \[\sum \Ot(d_v+1) \le \Ot\left(\sum d_v+n\right) \le \Ot(m).\]
\end{proof}

\section*{Acknowledgements and AI use}
We thank the anonymous ESA reviewers for their helpful comments.

The main \emph{vertex contraction} idea in the algorithm behind~\Cref{thm:additional_arc_removal} was devised by ChatGPT 5.5 Extended when queried for the known $\Ot(n)$-depth parallel algorithms for converting a \emph{preflow}~\cite{DBLP:journals/jacm/GoldbergT88} into a feasible flow. The simpler problem studied in this section can be reduced to preflow conversion.
Based on the vertex contraction idea, the authors devised an $\Ot(m)$-work algorithm described in~\Cref{sec:removing_v_s_and_t_v}.
To the best of our knowledge, an $\Ot(n)$-depth preflow coversion algorithm has not been described in the prior literature.

Codex with ChatGPT 5.6 Sol (Extra High) was used at the final stage of manuscript preparation for proofreading, grammar checks, and stylistic improvements.

\bibliographystyle{alpha}
\bibliography{references}
\clearpage
\appendix

\section{Omitted proofs}\label{sec:omitted}

\begin{lemma}\label{lem:basic_flow}
   Suppose a feasible flow satisfies item~\ref{basic_forest}~of~\Cref{def:basic_flow}. Then it also satisfies item~\ref{basic_convex}~of~\Cref{def:basic_flow}.
\end{lemma}
\begin{proof}
   Assume the contrary: that we have $f = (1 - \alpha)f' + \alpha f''$ for some feasible flows $f' \neq f''$ and $\alpha\in(0,1)$, and the set of arcs $e$ such that $0 < f_e < u_e$ forms a forest without a path between $s$ and~$t$.

   Consider some $e\in E$. Observe that if $f_e = u_e$ or $f_e = 0$ then $f'_e = f''_e$. Indeed, $f_e = u_e$ can only hold if $f'_e = f''_e = u_e$ and $f_e = 0$ is true only if $f'_e = f''_e = 0$. Consequently, the multiset of undirected arcs obtained from the arcs~$e$ such that $f'_e \neq f''_e$ is a subforest of the forest formed by the arcs satisfying $0 < f_e < u_e$. Since $f'\neq f''$, this subforest has at least one non-trivial component, that is, one containing at least two vertices.

   Now let us consider a leaf $v\in V$ of a non-trivial component of this subforest. $v$ has exactly one arc $e$ for which $f'_e \neq f''_e$, and it is either incoming or outgoing.
   Because the flows $f'$ and $f''$ are equal on all other arcs incident to $v$, we have $\ex_{f'}(v)\neq \ex_{f''}(v)$.
   As a result, one of these excess values is not zero, so $v\in\{s,t\}$, as otherwise either $f'$ or $f''$ would not be a feasible flow.
   We thus conclude that all the leaves in the subforest are in $\{s,t\}$.
   This means that the only non-trivial component in the subforest is a path from $s$ to $t$.
   This contradicts the assumption that there is no such path in the original forest.
\end{proof}

\section{Counterexample to the algorithm of \cite{DBLP:journals/jpdc/PeretzF22}}\label{sec:error}
Our counterexample will show that \cite[Algorithm~1]{DBLP:journals/jpdc/PeretzF22} can perform $\Theta(n^2)$ iterations for some instances, meaning its depth is not $O(n)$.

The counterexample, illustrated in \Cref{figure:peretz}, consists of vertices $s, v, t, k_1, \dots, k_n, l_1, \dots, l_n$. We add infinite-capacity arcs $(s, k_i)$ and $(l_i, v)$ for $i = 1, \dots, n$, an infinite-capacity arc $(v, t)$, and $n^2$ unit-capacity arcs $(k_i, l_j)$ for $i \in \{1, \dots, n\}$ and $j \in \{1, \dots, n\}$. This instance has a maximum flow of $n^2$.
\begin{figure}[h!]
   \begin{center}
      \begin{tikzpicture}[style=tikzfig, scale=0.998]
      	\begin{pgfonlayer}{nodelayer}
      		\node [style=peretz] (0) at (-4, -0.5) {$s$};
      		\node [style=peretz] (1) at (0, 4) {$k_1$};
      		\node [style=peretz] (2) at (0, 1) {$k_2$};
      		\node [style=peretz] (3) at (0, -2) {$k_{n-1}$};
      		\node [style=peretz] (4) at (0, -5) {$k_n$};
      		\node [style=none] (dots) at (0, -0.5) {$\cdots$};
      		\node [style=peretz] (5) at (4, 4) {$l_1$};
      		\node [style=peretz] (6) at (4, 1) {$l_2$};
      		\node [style=none] (dots) at (4, -0.5) {$\cdots$};
      		\node [style=peretz] (7) at (4, -2) {$l_{n-1}$};
      		\node [style=peretz] (8) at (4, -5) {$l_n$};
      		\node [style=peretz] (9) at (8, -0.5) {$v$};
      		\node [style=peretz] (10) at (12, -0.5) {$t$};
      	\end{pgfonlayer}
      	\begin{pgfonlayer}{edgelayer}
      		\draw [style=default diedge] (4) to (7);
      		\draw [style=default diedge] (3) to (7);
      		\draw [style=default diedge] (2) to (7);
      		\draw [style=default diedge] (1) to (7);
      		\draw [style=default diedge] (4) -- node[below]{$1$} (8);
      		\draw [style=default diedge] (3) to (8);
      		\draw [style=default diedge] (2) to (8);
      		\draw [style=default diedge] (1) to (8);
      		\draw [style=default diedge] (2) to (6);
      		\draw [style=default diedge] (2) to (5);
      		\draw [style=default diedge] (1) to (6);
      		\draw [style=default diedge] (1) -- node[above]{$1$} (5);
      		\draw [style=default diedge] (3) to (6);
      		\draw [style=default diedge] (4) to (6);
      		\draw [style=default diedge] (3) to (5);
      		\draw [style=default diedge] (4) to (5);
      		\draw [style=default diedge] (0) -- node[above, sloped]{$\infty$} (1);
      		\draw [style=default diedge] (0) -- node[above, sloped]{$\infty$} (2);
      		\draw [style=default diedge] (0) -- node[above, sloped]{$\infty$} (3);
      		\draw [style=default diedge] (0) -- node[above, sloped]{$\infty$} (4);
      		\draw [style=default diedge] (5) -- node[above, sloped]{$\infty$} (9);
      		\draw [style=default diedge] (6) -- node[above, sloped]{$\infty$} (9);
      		\draw [style=default diedge] (7) -- node[above, sloped]{$\infty$} (9);
      		\draw [style=default diedge] (8) -- node[above, sloped]{$\infty$} (9);
      		\draw [style=default diedge] (9) -- node[above, sloped]{$\infty$} (10);
      	\end{pgfonlayer}
      \end{tikzpicture}
   \end{center}
   \caption{Counterexample}\label{figure:peretz}
\end{figure}
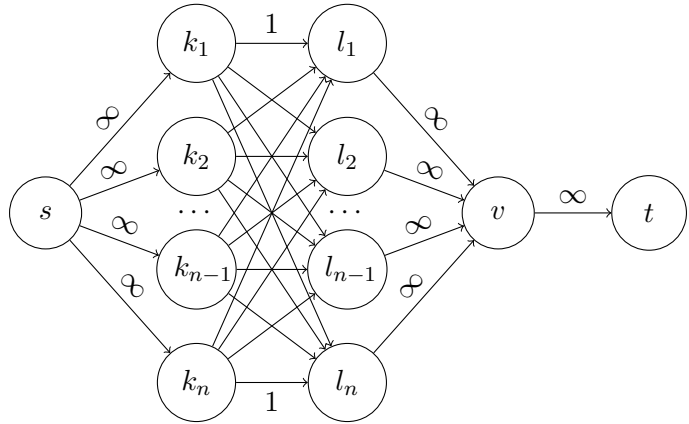

\cite[Algorithm~1]{DBLP:journals/jpdc/PeretzF22} first splits the sink $t$ into multiple sinks $t_1,\ldots,t_k$, one for each of the incoming arcs of~$t$.
This step in our instance has no particular effect.

Then, the algorithm proceeds in iterations.
The algorithm then, in each iteration, constructs a directed subtree $T$ of the residual network rooted at $s$  and computes a flow in $T$ from $s$ to the sinks $t_1,\ldots,t_k$, maximizing the sum of flows sent to these sinks. In our instance there is only one sink vertex, so the flow is pushed via a single path each time.

Since the path from $s$ to $t$ has to go through an arc $(k_i, l_j)$ with capacity 1, the algorithm will push exactly 1 unit of flow.
In the residual network, pushing 1 unit of flow through an arc with capacity 1 reverses it, maintaining the invariant that all arcs $(k_i, l_j)$ will either have capacity 1 or 0. This means that every iteration can push at most 1 unit of flow, which in turn implies that the algorithm has to perform at least $n^2$ iterations to find the maximum flow.

\end{document}